\documentclass[sts,preprint]{imsart}

\usepackage{amsthm,amsmath,amsfonts,amssymb}

\usepackage[authoryear]{natbib}
\usepackage[colorlinks,citecolor=blue,urlcolor=blue]{hyperref}
\usepackage{graphicx}
\usepackage{enumerate}
\usepackage{dsfont}
\usepackage{cleveref}
\usepackage{tikzit}

\tikzstyle{myvar}=[fill=none, draw=none, shape=circle, tikzit draw=none, tikzit shape=circle]

\tikzstyle{arrow}=[draw=black, ->, tikzit draw=black]
\tikzstyle{probedge}=[draw=black, -, tikzit draw = black, dashed]

\startlocaldefs
\usetikzlibrary{patterns, arrows.meta}

\usepackage{enumitem}

\newcommand{\E}{\mathbb{E}}
\newcommand{\Prob}{\mathbb{P}}
\newcommand{\V}{\mathbb{V}}

\newcommand{\mathspace}[1] {\mathcal{#1}}

\newcommand{\vectorobs}[1] {\mathrm{#1}}
\newcommand{\scalarobs}[1] {#1}
\newcommand{\RV}[1] {#1}

\newcommand{\Xspace}{\mathspace{X}}
\newcommand{\rvX}{\RV{X}}
\newcommand{\rvY}{\RV{Y}}
\newcommand{\rvZ}{\RV{Z}}

\newcommand{\rvN}{\RV{N}}

\newcommand{\obsZ}{\vectorobs{z}}

\newcommand{\obsN}{\scalarobs{n}}
\endlocaldefs

 \def\firstcircle{(90:1.25cm) circle (1.5cm)}
  \def\secondcircle{(210:1.25cm) circle (1.5cm)}
  \def\thirdcircle{(330:1.25cm) circle (1.5cm)}

\newcommand\independent{\protect\mathpalette{\protect\independenT}{\perp}}
\def\independenT#1#2{\mathrel{\rlap{$#1#2$}\mkern2mu{#1#2}}}

\newtheorem{theorem}{Theorem}
\newtheorem{corollary}{Corollary}
\newtheorem{lemma}{Lemma}

\newtheorem{definition}{Definition}
\begin{document}

\begin{frontmatter}
\title{Feature selection in stratification estimators of causal effects: lessons from potential outcomes, causal diagrams, and structural equations}
\runtitle{Feature selection in stratification estimators of causal effects}

\begin{aug}
\author[A]{\fnms{P. Richard} \snm{Hahn}\ead[label=e2]{prhahn@asu.edu}}
\and
\author[B]{\fnms{Andrew} \snm{Herren}\ead[label=e1]{asherren@asu.edu}}
\affiliation{Arizona State University}
\end{aug}

\begin{abstract}
What is the ideal regression (if any) for estimating average causal effects? We study this question in the setting of discrete covariates, deriving expressions for the finite-sample variance of various stratification estimators. This approach clarifies the fundamental statistical phenomena underlying many widely-cited results. Our exposition combines insights from three distinct methodological traditions for studying causal effect estimation: potential outcomes, causal diagrams, and structural models with additive errors.  
\end{abstract}


\end{frontmatter}

\section{Introduction}

This paper considers the problem of estimating an average treatment effect from observational or experimental data, provided that a sufficient set of control variables are available. We pose the question: might the statistical precision of our estimates improve if we used only a subset of the available controls or possibly a dimension reduced transformation of them? This question is evergreen in the applied social sciences (see \cite{leamer1983let} or \cite{hernan2020causal}, page 195), but is surprisingly tricky to navigate for many applied researchers. In this paper, we break the problem down by considering the somewhat stylized situation of discrete covariates with finite support, where we are able to conduct a thorough variance analysis.

This paper examines this question in detail using tools from three distinct formalisms: potential outcomes, causal diagrams, and structural equations. We show (Section 2) that a key condition licensing valid causal inference from observational data can be expressed equivalently in each of the three distinct frameworks (conditional unconfoundedness, the back-door criterion, and exogenous errors), allowing us to alternate between perspectives as is convenient pedagogically. Importantly, this equivalence is established in terms of a generic function of observed covariates, meaning that it covers not only variable selection, but ``feature selection''; this generality means that insights built on this equivalence apply seamlessly to modern methods such as regression trees or neural networks, which implicitly introduce potentially non-invertible transformations of the observed covariates.

For clarity, we focus on the simplified (yet fairly common in practice) setting of discrete covariates with finite support, which allows us to derive finite sample properties of common stratification estimators, including widely-used linear regression and propensity score methods. Section 3 presents two novel-but-elementary results that will be used to re-analyze earlier theoretical results pertaining to regression adjustment for causal effect estimation. The first result defines the notion of a minimal control function, allowing us to distinguish between necessary and sufficient statistical control for causal effect estimation. The second result is a finite-sample analysis of stratification estimators of average causal effects in the setting of discrete control variables with finite support. This finite-sample analysis, presented in Theorem \ref{theorem2}, articulates the conditions by which a control function may be viewed as optimal in the sense of minimum variance.

Section 4 collects concrete examples illustrating practical implications of the theory presented in Section 3, detailing how these results relate to previous literature, both classic and contemporary. By bringing together these profound results in the context of a common statistical framework, we hope to harmonize their insights for practitioners. 
 
 Section 5 concludes by discussing further connections to previous literature. 
 
\section{Formal frameworks for causal inference}
Let $\rvY$ be the outcome/response of interest, $\rvZ$ be a binary treatment assignment, and $\rvX$ be a vector of covariates drawn from covariate space $\Xspace$, all denoted here as random variables. For a sample of size $n$, observations are assumed to be drawn independently as triples $(X_i, Y_i, Z_i)$, for $i = 1, \dots, n$. The goal of causal effect estimation is to understand how the response variable $Y$ changes according to hypothetical manipulations of the treatment assignment variable, $Z$. 
For simplicity, we will refer to our observational units as ``individuals'', although of course in applications that need not be the case.

The essential challenge to causal estimation is that only one of the two possible treatment assignments can be observed; as a consequence, if individuals who happen to receive the treatment differ systematically from those who do not, either in terms of their likely response value or in terms of how they respond to treatment, naive comparisons between the treated and untreated units will not simply reflect the causal impact of the treatment --- the treatment effect is said to be {\em confounded} with other aspects of the population. The field of causal inference has proposed and developed a variety of techniques for coping with this difficulty, the most common of which is some form of regression adjustment (meant here to include propensity score estimators and matching estimators, etc), which entails estimating average causal effects as (weighted) averages of (estimated) conditional expectations. The key assumption that justifies this process is referred to as {\em conditional unconfoundedness}, which asserts that the measured covariates adequately account for all of the systematic differences between the treated and untreated individuals in our observational sample; formalizing this assumption can be approached in a number of ways, which we turn to now. Only after the notation of these formalisms has been introduced can our causal estimand, and the class of estimators we will study, be precisely defined.

\subsection{Potential outcomes} \label{potential_outcome_section}
The potential outcomes framework casts causal inference as a missing data problem: causal estimands are contrasts between pairs of outcomes that are mutually unobservable --- when we see one, we cannot see the other. At present, the standard reference for the potential outcomes framework is \cite{imbens2015causal}, which contains extensive citations to the primary literature.

Let $\rvY^1$ and $\rvY^0$ refer to the ``potential outcomes'' when $\rvZ=1$ and $\rvZ=0$. For individual $i$, the {\em individual treatment effect} will be defined as the difference between the potential outcomes: $$\tau_i = Y^1_i - Y^0_i.$$ Other treatment effects, such as a ratio rather than a difference, are sometimes considered, but in this paper we focus on the difference. Because the potential outcomes $(\rvY^1, \rvY^0)$ are never observed simultaneously, individual treatment effects can never be estimated directly. 

However, {\em average} treatment effects can be identified (learned from data) provided certain assumptions are satisfied. The causal estimand this paper will focus on is the average treatment effect, or ATE:
\begin{equation}
\bar{\tau} \equiv \E[\rvY^1 - \rvY^0].
\end{equation}
The precise population over which this expectation is taken will be discussed in more detail in section \ref{estimands}. The standard assumptions that allow this average effect to be estimated are:
\begin{enumerate}
\item Stable unit treatment value assumption (SUTVA), which consists of two conditions:
\begin{enumerate}
\item {\em Consistency}: The observed data is related to the potential outcomes via the identity 
\begin{equation}\label{gating}
\rvY = \rvY^1 Z + \rvY^0 (1 - Z),
\end{equation}
which describes the ``gating'' role of the observed treatment assignment, $Z$.
\item {\em No Interference}: for any sample of size $\obsN$ with $\rvY \in \mathcal{Y}$ 
and $\rvZ \in \mathcal{Z}$, $(\rvY_i^1, \rvY_i^0) \independent \rvZ_j$ for all $i,j \in \{1, ..., \obsN\}$ with $j \neq i$, which rules out interference between observational units.
\end{enumerate}
\item Positivity: $0 < \Prob(\rvZ=1 \mid \rvX= x) < 1$ for all $x \in \mathcal{X}$
\item Conditional unconfoundedness: $(\rvY^1, \rvY^0) \independent \rvZ \mid \rvX$
\end{enumerate}
Imagining concrete violations of these conditions is intuition-building. Consistency can be violated under non-compliance, so that treatment assignment doesn't match treatment actually received. No interference can be violated, for example, if we were studying the effect of individual tutoring on student grades in a certain classroom and students study together; Jimmy's treatment assignment may impact Sally's grade. Positivity is violated if certain individuals can never receive treatment, rendering their contribution to the average treatment effect unlearnable. And finally, conditional unconfoundedness can be violated, for example, if both treatment assignment and the outcome variable share a common cause. However, this is not the only way conditional unconfoundedness can be violated, and exploring other possibilities in full generality is the topic of the remainder of the paper. 

Taken together, the above assumptions enable identification of average treatment effects because they imply the following equality, the left-hand side of which is estimable: 
\begin{equation*}
\begin{aligned}
\E_X[\E[\rvY \mid \rvX, \rvZ = 1] - \E[\rvY \mid \rvX, \rvZ = 0]] & =  \E[\rvY^1 - \rvY^0].
\end{aligned}
\end{equation*}
In more detail, the equivalence is established as follows:
\begin{equation*}
\begin{aligned}
\E_X[\E[\rvY \mid \rvX, \rvZ = 1] &= \E_X[\E[\rvY^1 Z + \rvY^0 (1-Z) \mid \rvX, \rvZ = 1]]\\
&= \E_X[\E[\rvY^1 \mid \rvX, \rvZ = 1] = \E[\rvY^1].\\
\E_X[\E[\rvY \mid \rvX, \rvZ = 0] &= \E_X[\E[\rvY^1 Z + \rvY^0 (1-Z) \mid \rvX, \rvZ = 0]]\\
&= \E_X[\E[\rvY^0 \mid \rvX, \rvZ = 0] = \E[\rvY^0].
\end{aligned}
\end{equation*}
An alternative parametrization is: $$Y_i = Y_i^0 + \tau_i Z_i$$
where $$\tau_i = Y_i^1 - Y_i^0,$$ which emphasizes that $\tau_i$ itself can differ across units 
and, as a random variable, can be {\em dependent} on the treatment assignment so that $\tau \not\independent Z$. This treatment effect parametrization will be used extensively in our exposition.

This paper is focused on the following question: If $X$ satisfies conditional unconfoundedness, might there be a function of $X$ with a reduced range that also satisfies conditional unconfoundedness? That is, can $X$ be reduced in dimension while still providing valid causal effect estimation? Answering this question requires a more detailed examination of {\em how} conditional unconfoundedness is achieved in any particular data generating process, which is facilitated by the introduction of causal diagrams.

\subsection{Causal diagrams} \label{DAG_section}

\subsubsection{Graph theory for causal identification}

Causal diagrams provide a more fine-grained look at confounding, as they consider the full joint distribution of the response, treatment, and control variables regressors. The graphical approach to causality has its earliest roots in the work of Sewell Wright \citep{wright1918nature, wright1920relative, wright1921correlation}, but attained its mature modern form in the prodigious work of Judea Pearl  \citep{pearl1987embracing, pearl1987logic, pearl1995theory, pearl1995causal}. See \cite{pearl2009causality} for a textbook treatment and comprehensive references. The presentation here loosely follows the expository treatment in \cite{shalizi2021advanced}. 

Recall that any joint density over $p$ random variables may be expressed in {\em compositional form}, as a product of conditional densities:
$$f(x_1, x_2, \dots, x_p) = f(x_1)f(x_2 \mid x_1)f(x_3 \mid x_1, x_2)...f(x_p \mid x_1, x_2, \dots, x_{p-1}),$$
where the density functions $f(\cdot)$ and $f(\cdot \mid \cdot)$ refer to different densities depending on their arguments. 
The labeling of the variables is arbitrary, and so we can chain together these marginal and conditional distributions in any order (though of course that will lead to different forms).
Some of these variables might exhibit {\em conditional independence}, meaning that, for example
$$f(x_1 \mid x_2, x_3) = f(x_1 \mid x_2)$$
which is equivalently expressed as $$X_1 \independent X_3 \mid X_2.$$  The relationship to {\em directed (acyclic) graphs} (DAG) is straightforward: draw a node for each variable and draw a line from $X_j$ going into $X_i$ if $X_j$ appears in the conditional distribution of $X_i$. 
This graph is {\em directed}, with the arrow pointing from $X_j$ {\em to} $X_i$. We say that $X_j$ is a ``parent'' of $X_i$ and that $X_i$ is the ``child'' of $X_j$. 

From the graph, the joint distribution may be expressed as
$$f(x_1, \dots, x_p) = \prod_{j = 1}^p f(x_j \mid \mbox{parents}(x_j)).$$ 
This leads us to the {\em Markov property}, which is
$$X_j \independent \mbox{non-descendants}(X_j) \mid \mbox{parents}(X_j),$$
where ``descendant'' refers to children, grandchildren, great-grandchildren, etc. We can see this by dividing through by the marginal distribution of 
$\mbox{parents}(X_j)$ and observing that the resulting distribution is a product of terms involving either $X_j$ or $\mbox{non-descendants}(X_j)$, but not both. The Markov property allows one to efficiently deduce conditional independence relationships and underpins Pearl's algorithm (which will be described shortly).

Finally, a complete treatment of confounding in the causal diagram framework requires the following definition: 
\begin{definition}
A {\em collider} is a node/variable $V$ in a DAG that sits on an undirected path between two other nodes/variables, $X_j$ and $X_i$, and the paths both have arrows pointing {\em into} $V$. 
\end{definition}
Conditioning on a collider induces dependence between its parents. For a classic example of this phenomenon, suppose that a certain college grants admission only to applicants with high test scores and/or athletic talent. Even if we grant that in the general population these talents may be independent, but among admitted students, these two attributes become highly dependent.  If we know that a student is not athletic, then we know for sure that they must be academically gifted and vice-versa. While this is a basic result in probability theory, Pearl's work emphasized its significance to the problem of regression adjustment for causal effect estimation.

With a DAG in hand, it is possible to deduce -- rather than assume --  conditional unconfoundedness: Pearl developed an algorithm for determining subsets of variables in $X$ (i.e., its coordinate dimensions) that define valid regression estimators. The inputs to this algorithm are a directed acyclic graph (DAG) that characterizes the causal relationships between variables; such a graph describes a particular compositional representation of the joint distribution, reflecting conditional independences that are implied by the {\em stipulated} causal relationships. The prohibition on cycles rules out positive feedback self-causation. Here we present Pearl's algorithm in a somewhat simplified form, assuming that the graph contains no descendants of $Z$ other than $Y$.

Given an input DAG, $\mathcal{G}$ and a subset of nodes $S$, the ``backdoor'' algorithm proceeds as follows:
\begin{enumerate}
\item Identify all (undirected) paths between $Z$ and $Y$. 
\item Consider each variable along each of these paths and make sure that at least one of them is ``blocked''.
\begin{enumerate}
\item A variable $W$ is blocked if 
\begin{enumerate}
\item $W$ is not a collider and is in the set $S$ or
\item $W$ is a collider and neither $W$ nor any of its descendants is in the set $S$.
\end{enumerate}
\end{enumerate}
\item Return {\tt TRUE} if every ``backdoor'' path between $Z$ and $Y$ (all paths except the direct causal arrow from $Z$ to $Y$), is blocked. Otherwise return {\tt FALSE}.
\end{enumerate}
Sets of variables satisfying the backdoor criterion --- those sets where the algorithm returns {\tt TRUE} --- are valid adjustment sets in the sense that $Y$ and $Z$ {\em would be} conditionally independent, given those variables, {\em if there were no causal relationship} between $Y$ and $Z$. By ruling out all other possible sources of association, any observed association may be interpreted as arising from a causal relationship.

\subsubsection{Functional causal models.}\label{fcm}
Causal DAGs may be associated with a functional causal model, 
a set of deterministic functions that take as inputs elements of $X$ as well as independent (``exogenous'') error terms.  
The basic triangle confounding graph corresponding to an $(X, Y, Z)$ triple satisfying conditional unconfoundedness is shown in Figure \ref{graph1}.
\begin{figure}
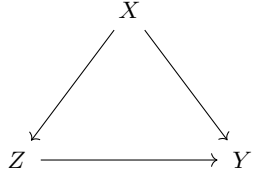

\ctikzfig{graph1}
\caption{A simple triangle confounding diagram, where a control variable $X$ causally influences both the treatment $Z$ and the response $Y$. This graph does not clarify what information contained in (the potentially multidimensional) $X$ is relevant for $Z$ or $Y$ or both or neither, only that knowing the value of $X$ in its entirety permits causal estimation.}\label{graph1}
\end{figure}
The corresponding functional causal model can be expressed as
\begin{equation}
\begin{split}
Z &\leftarrow G(X,\epsilon_z)\\
Y & \leftarrow F(X,Z,\epsilon_y)
\end{split}
\end{equation}
where $X$, $\epsilon_z$ and $\epsilon_y$ are mutually independent (though all three may be vector-valued with non-independent elements). The exogenous errors ($\epsilon_z$ and $\epsilon_y$) that appear in a single equation are suppressed in the graph. All of the stochasticity is inherited from the exogenous variables, while all of the deterministic relationships are reflected in the functions $G(\cdot)$ and $F(\cdot)$, which are explicitly endowed with a causal interpretation. Specifically, the potential outcomes are given by:
\begin{equation}
\begin{split}
Y^1 &\leftarrow F(X,1,\epsilon_y)\\
Y^0 & \leftarrow F(X,0,\epsilon_y)
\end{split}
\end{equation}
where $(X, \epsilon_y)$ are drawn from their marginal distributions, irrespective of the value of the treatment argument. As was mentioned previously, throughout this paper we assume that $X$ does not contain any causal descendants of $Z$.

Consider two ways to conceptualize the data generating process for both the potential outcome pairs, $(Y^0, Y^1)$, and the observed response $Y$. On the one hand, the potential outcomes can be generated from the functional causal model, by fixing the $Z$ argument to 0 or 1, irrespective of the implied distribution of $Z \mid X$. Procedurally, this would look like drawing $X$ from its marginal distribution, drawing $\epsilon_y$, and evaluating $F(X, 0, \epsilon_y)$ and $F(X, 1, \epsilon_y)$. The observed data can then be constructed via the consistency assumption $Y = F(X, 1, \epsilon_y)Z + F(X, 0, \epsilon_y)(1-Z)$. Equivalently, $Y$ may be drawn directly via $F(X, Z, \epsilon_y)$, where $Z$ (the observed treatment assignment) was drawn according to $Z \mid X$ (as specified by the CDAG). This equivalence is especially instructive as to why $Y \mid Z = z$ and $Y^z$ do not generally have the same distribution and, furthermore, why $Y \mid Z = z, X = x$ and $Y^z \mid X= x$ do have the same distribution (assuming, as we have above, that $X$ is causally exhaustive).

The role of $\epsilon_y$ in defining the distribution of the potential outcomes is worth considering in more detail. Note that for a binary $Z$, any functional causal model $F$ may be rewritten as
$$F(X, Z , \epsilon_y) = F(X, 0, \epsilon_y) + Z \left[ F(X, 1, \epsilon_y) - F(X, 0, \epsilon_y) \right] = \mu(X, \epsilon_y) + Z \tau(X, \epsilon_y).$$ This formulation invites us to consider that $\epsilon_y$ may be multivariate, distinct elements of which may affect $\mu(X, \epsilon_y)$ and $\tau(X, \epsilon_y)$.  Three 
particular cases are especially notable:
\begin{enumerate}
\item $\mu(X, \epsilon_y) = \mu(X) + \epsilon_y$ and $\tau(X, \epsilon_y) = \tau(X)$: here, $\epsilon_y$ has the same effect on the two potential outcomes $F(X, 1, \epsilon_y)$ and $F(X, 0, \epsilon_y)$, so that their joint distribution is singular.
\item $\mu(X, \epsilon_y) = \mu(X) + \epsilon_{y,0}$ and $\tau(X, \epsilon_y) = \tau(X) + \left( \epsilon_{y,1} - \epsilon_{y,0} \right)$ where the exogenous error is partitioned as $\epsilon_y = (\epsilon_{y, 0}, \epsilon_{y, 1})$. Here, $\epsilon_{y,0}$ and $\epsilon_{y,1}$ are distinct random variables that separately define the potential outcome distributions so that one effect of the treatment is in changing {\em which} exogenous influences affect the response. 
\item $\mu(X, \epsilon_y) = \mu(X) + \epsilon_{y, \mu}$, $\tau(X, \epsilon_y) = \tau(X) + \epsilon_{y, \tau}$, where the exogenous errors is partitioned as $\epsilon_y = (\epsilon_{y, \mu}, \epsilon_{y, \tau})$. In this case, a distinct set of causal factors dictate exogenous variation in the prognostic (baseline) response and exogenous variation in the treatment effect itself. For example, variation in the baseline response may be due to environmental factors that are independent from genetic factors dictating one's response to a new drug.
\end{enumerate} 
These three cases are visualized in Figure \ref{errors} with $\tau(X) = 1$. Empirically, these cases are indistinguishable in that they are ``observationally equivalent'' --- because the potential outcomes are never jointly observed, most aspects of their joint distribution are fundamentally unidentified. 

\begin{figure}
\includegraphics[width=2.5in]{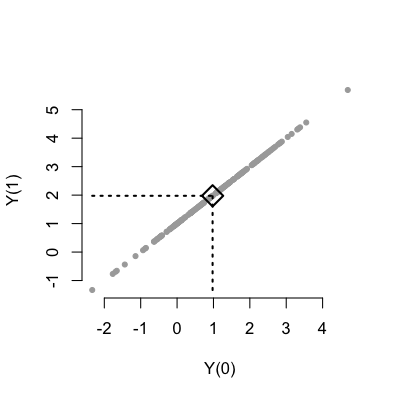}
\includegraphics[width=2.5in]{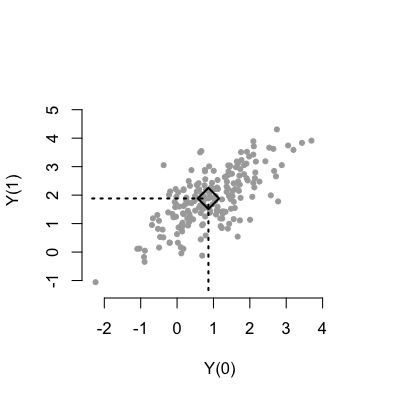}
\caption{Left panel: Potential outcome distributions with a common additive univariate error and a homogeneous treatment effect (which shifts the line up one unit from the diagonal), articulated in Case 1 below. Right panel: Potential outcome distributions with a homogeneous treatment effect and distinct additive bivariate errors, $\epsilon_{y,0}$ and $\epsilon_{y,0}$, shown here with a positive correlation less than one, articulated in Case 2 below.}\label{errors}
\end{figure}

With a more detailed causal graph, a more detailed assessment of conditional unconfoundedness can be made. For instance, consider Figure \ref{graph2}, which is equivalent to the standard triangle digram in the sense that controlling for all of the elements of $X = (X_1, X_2, X_3, X_4)$ indeed satisfies conditional unconfoundedness. However, Pearl's algorithm reveals that $(X_1, X_2)$ would suffice. By positing more information about the joint distribution of $X$, it is possible to absorb $X_3$ into $\epsilon_z$ and $X_4$ into $\epsilon_y$, while redefining $X = (X_1, X_2)$, bringing us back to the triangle graph, but with a reduced set of control variables.
\begin{figure}
\centering
\begin{tikzpicture}[baseline=-0.25em,scale=0.5]
	\begin{pgfonlayer}{nodelayer}
		\node [style=myvar] (1) at (2.5, 2.5) {$X_1$};
		\node [style=myvar] (2) at (2.5, -2.5) {$X_2$};
		\node [style=myvar] (3) at (-2.5, -2.5) {$X_3$};
		\node [style=myvar] (4) at (7.5, 2.5) {$X_4$};
		\node [style=myvar] (5) at (-2.5, 0) {$Z$};
		\node [style=myvar] (6) at (7.5, 0) {$Y$};
	\end{pgfonlayer}
	\begin{pgfonlayer}{edgelayer}
		\draw [style=arrow] (1) to (5);
		\draw [style=arrow] (1) to (6);
		\draw [style=arrow] (2) to (5);
		\draw [style=arrow] (2) to (6);
		\draw [style=arrow] (3) to (5);
		\draw [style=arrow] (4) to (6);
		\draw [style=arrow] (5) to (6);
	\end{pgfonlayer}
\end{tikzpicture}
\caption{An elaboration of the triangle graph, depicting $X_1$ and $X_2$ as confounders, $X_4$ as a pure prognostic variable, and $X_3$ is an instrument.}\label{graph2}
\end{figure}
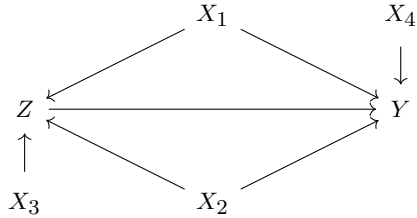

\subsection{Structural equations: Mean regression models with exogenous additive errors} \label{structural_model_section}

Finally, the classic econometric literature approaches causality in terms of mean regression models with additive (but not necessarily homoskedastic) error terms, which are referred to as ``structural'' models (although the term is often used informally and imprecisely in the applied literature).  \cite{heckman2005structural} reviews the structural model approach in econometrics in depth, noting that such methods have their origin in the study of dynamic macroeconomic systems. 
A seminal reference is \cite{haavelmo1943statistical}. The mean regression perspective arises naturally if one takes a linear regression model as a starting point, but is straightforward to motivate starting from a generic functional causal model. 

Define 
\begin{equation}
\begin{split}
\mu(x) &\equiv \E(F(x,0,\epsilon_y)), \\
\tau(x) &\equiv \E(F(x,1,\epsilon_y)) - \mu(x),\\
\upsilon(x,\epsilon_y) &\equiv F(x,0,\epsilon_y) - \mu(x),\\
\delta(x,\epsilon_y) &\equiv F(x,1,\epsilon_y) - F(x,0,\epsilon_y) - \tau(x)
\end{split}
\end{equation}
giving a ``structural model'' 
\begin{equation}\label{structural_eq}
Y = \mu(x) + \upsilon(x,\epsilon_y) + (\tau(x) + \delta(x,\epsilon_y)) Z
\end{equation}
where $\upsilon(x, \epsilon_y)$ and $\delta(x, \epsilon_y)$ are deterministic functions, both of which are mean zero integrating over $\epsilon_y$ (for any $x$): $\E(\upsilon(x, \epsilon_y)) = 0$ and $\E(\delta(x, \epsilon_y)) = 0$. In this formulation, conditional unconfoundedness may be expressed in terms of independence of the treatment, $Z$, and the error terms $\upsilon(x,\epsilon_y)$ and $\delta(x, \epsilon_y)$. Such models are commonly used in a simplified form, where $\delta(x, \epsilon_y)$ is assumed to be identically zero and $\tau(x)$ is assumed to be constant in $x$, but such assumptions are not intrinsic to the formalism.

\subsection{Relating the three frameworks}\label{equivalence}

If every node in a causal diagram is observable, all remaining factors determining $Y$ are attributable to the exogenous errors, which are, by definition, independent of the treatment assignment. In that case, it is easy to forge a connection between the three formalisms, as they all assert that
\begin{equation}
Y^z \mid X=x \;\;\; \,{\buildrel d \over \sim}\, \;\;\; Y \mid X = x, Z = z,
\end{equation} 
where (recall) $Y^z = F(x, z, \epsilon_y)$, with distribution induced by the distribution over $\epsilon_y$. The above assertion  essentially declares that the estimable conditional distributions which appear on the right hand side warrant a causal interpretation. 

For sets of control variables that are {\em not} exhaustive, more care is needed in translating the formalisms, but a precise relationship can be obtained, as spelled out in the following lemma.
\begin{lemma}\label{synthesis}
The assertions below (with their corresponding causal framework labeled in brackets) stand in the following logical relationship: $1 \Rightarrow 2 \Leftrightarrow 3$.
\begin{enumerate}
\item $S = s(X)$ satisfies the back-door criterion. [Causal DAGs]
\item $S= s(X)$ satisfies conditional unconfoundedness: $(Y^0, Y^1) \independent Z \mid S$. [Potential Outcomes]
\item The response $Y$ can be represented in terms of a mean regression model with error terms $(\upsilon(s,X,\epsilon_y),  \delta(s,X, \epsilon_y) ) \independent Z \mid s(X) = s$. [Structural Equations]
\end{enumerate}

\end{lemma}
\begin{proof}
Let $X$  denote all of the variables in a complete causal diagram with the exception of the treatment variable $Z$ and response variable $Y,$ and consider the following causal model, written in terms of functional equations, potential outcomes, and a structural mean regression with additive exogenous errors:
\begin{equation}
\begin{split}
Z &\leftarrow G(X, \epsilon_z),\\
Y^z &\leftarrow F(X, z, \epsilon_y) =  \mu(X) + \upsilon(X,\epsilon_y) + (\tau(X) + \delta(X,\epsilon_y)) z,\\
\begin{pmatrix} Y^0 \\
 Y^1 \end{pmatrix} &\leftarrow  \begin{pmatrix} \mu(X) + \upsilon(X,\epsilon_y) \\ \mu(X) + \tau(X) +  \upsilon(X,\epsilon_y) + \delta(X, \epsilon_y) \end{pmatrix}.
\end{split}
\end{equation}

To see that 1 implies 2, recall that 1 means that $S$ renders the treatment and response conditionally independent in the modified DAG with no causal arrow between $Z$ and $Y$. But it is precisely such a graph that defines the relationship between $Z$ and the potential outcomes $Y^0 = F(X, 0, \epsilon_y)$ and $Y^1 = F(X, 1, \epsilon_y)$, as shown in Figure \ref{po_graph}.

To see that 2 and 3 are equivalent, re-parametrize the additive error model in terms of $S$, as follows:
\begin{equation}
\begin{split}
Y^z &\leftarrow \mu(s) +  \upsilon(s,X,\epsilon_y) + (\tau(s) + \delta(s, X, \epsilon_y))z\\
\mu(s) &\equiv \E(\mu(X) \mid S(X) = s)\\
\tau(s) &\equiv \E(\tau(X) \mid S(X) = s)\\
\upsilon(s,X,\epsilon_y) &\equiv  \mu(X) - \mu(s) + \upsilon(X,\epsilon_y)\\
 \delta(s,X,\epsilon_y) &\equiv \tau(X) - \tau(s) + \delta(X,\epsilon_y).
 \end{split}
 \end{equation}
For a fixed value of $s$, the mean terms $\mu(s)$ and $\tau(s)$ are constant, so that $(Y^0, Y^1)$ stands in a one-to-one relationship with $\upsilon(s,X,\epsilon_y)$ and $\delta(s,X,\epsilon_y)$; therefore if the former are independent of $Z$, then so must be the latter, and vice-versa.

\end{proof}

\begin{figure}
\centering
\begin{tikzpicture}[baseline=-0.25em,scale=0.5]
	\begin{pgfonlayer}{nodelayer}
		\node [style=myvar] (1) at (2.5, 2.5) {$X_1$};
		\node [style=myvar] (2) at (2.5, -2.5) {$X_2$};
		\node [style=myvar] (5) at (-2.5, 0) {$Z$};
		\node [style=myvar] (6) at (7.5, 0) {$Y$};
	\end{pgfonlayer}
	\begin{pgfonlayer}{edgelayer}
		\draw [style=arrow] (1) to (5);
		\draw [style=arrow] (1) to (6);
		\draw [style=arrow] (2) to (5);
		\draw [style=arrow] (2) to (6);
		\draw [style=arrow] (5) to (6);
	\end{pgfonlayer}
\end{tikzpicture}
\hfill
\begin{tikzpicture}[baseline=-0.25em,scale=0.5]
	\begin{pgfonlayer}{nodelayer}
		\node [style=myvar] (1) at (2.5, 2.5) {$X_1$};
		\node [style=myvar] (2) at (2.5, -2.5) {$X_2$};
		\node [style=myvar] (5) at (-2.5, 0) {$Z$};
		\node [style=myvar] (6) at (7.5, 0) {$Y^*$};
	\end{pgfonlayer}
	\begin{pgfonlayer}{edgelayer}
		\draw [style=arrow] (1) to (5);
		\draw [style=arrow] (1) to (6);
		\draw [style=arrow] (2) to (6);
		\draw [style=arrow] (2) to (5);
	\end{pgfonlayer}
\end{tikzpicture}
\caption{A typical causal DAG (CDAG) and its potential outcome counterpart, where $Y^* = (Y^0, Y^1)$.}\label{po_graph}
\end{figure}
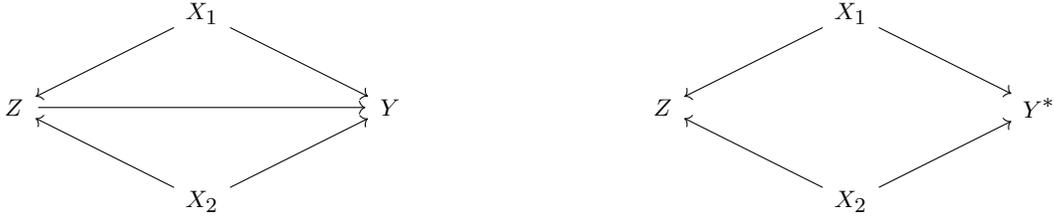

\begin{figure}
\begin{minipage}[b]{180pt}
\centering
\begin{tikzpicture}[baseline=-0.25em,scale=0.5]
	\begin{pgfonlayer}{nodelayer}
		\node [style=myvar] (1) at (-2.5, 2.5) {$X_1$};
		\node [style=myvar] (3) at (2.5, 2.5) {$X_2$};
		\node [style=myvar] (5) at (-2.5, 0) {$Z$};
		\node [style=myvar] (6) at (2.5, 0) {$Y$};
		\node [style=myvar] (7) at (-2.5, -2.5) {$X_3$};
		\node [style=myvar] (8) at (2.5, -2.5) {$X_4$};
	\end{pgfonlayer}
	\begin{pgfonlayer}{edgelayer}
		\draw [style=arrow] (1) to (5);
		\draw [style=arrow] (3) to (6);
		\draw [style=arrow] (5) to (6);
		\draw [style=arrow] (1) to (3);
		\draw [style=arrow] (7) to (5);
		\draw [style=arrow] (8) to (7);
		\draw [style=arrow] (8) to (6);
	\end{pgfonlayer}
\end{tikzpicture}
\caption{The ``box diagram'', which implies several valid control sets: any set containing at least one of $\{ X_1, X_2\}$ and at least one of $\{ X_3, X_4\}$.}
\label{rectangle}
\end{minipage}
\hfill
\begin{minipage}[b]{180pt}
\centering
\begin{tikzpicture}[baseline=-0.25em,scale=0.5]
	\begin{pgfonlayer}{nodelayer}
		\node [style=myvar] (1) at (2.5, 2.5) {$X_2$};
		\node [style=myvar] (2) at (2.5, -2.5) {$X_3$};
		\node [style=myvar] (5) at (-2.5, 0) {$Z$};
		\node [style=myvar] (6) at (7.5, 0) {$Y$};
	\end{pgfonlayer}
	\begin{pgfonlayer}{edgelayer}
		\draw [style=probedge] (1) to (5);
		\draw [style=arrow] (1) to (6);
		\draw [style=arrow] (2) to (5);
		\draw [style=arrow] (5) to (6);
		\draw [style=probedge] (2) to (6);
	\end{pgfonlayer}
\end{tikzpicture}
\caption{The ``box diagram'' with $X_1$ and $X_4$ omitted; a CDAG representation is no longer possible.\vspace{0.2in}}
\label{triangle}
\end{minipage}
\end{figure}

\subsection{Estimands, estimators, and sampling distributions}\label{estimands}
As described previously, by {\em treatment effect}, we mean the difference between the treated and untreated potential outcomes. By {\em average} treatment effect, we mean the average of this difference over some population of individuals. The functional causal model and a distribution over the exogeneous errors define an infinite hypothetical {\em population} from which the observed data is assumed to be a random sample. From this perspective, the population average treatment effect (PATE) may be expressed as $$\E(\tau(X) + \delta(X, \epsilon)) = \E(\tau(X)),$$ where $\tau$ is a fixed-but-unknown function and the expectation is taken with respect to the data generating process defined by the CDAG and the associated functional causal model, so that $X$ and $\epsilon$ are both being averaged over. 

Other average causal effects, differing in terms of the (sub)population over which the average is taken, are likewise readily defined in terms of the functional causal model (FCM). For instance, if we wish to restrict our attention to the average treatment effect among individuals in our observed sample, we may define our estimand as the {\em sample average treatment effect}, or SATE: $$\frac{1}{N} \sum_{i = 1}^{N} \left( \tau(x_i) + \delta(x_i, \epsilon_i) \right ).$$ Note that the SATE and the PATE differ from one another in that, in general,

$$\E(\tau(X)) \neq \frac{1}{N} \sum_{i = 1}^{N} \tau(x_i)$$
and 
$$\frac{1}{N} \sum_{i = 1}^{N} \delta(x_i, \epsilon_i) \neq \E(\delta(X, \epsilon)) = 0.$$
In this paper, we will compare stratification estimators of the PATE, evaluating them in terms of their finite sample variance over repeated sampling of independent draws from $(X_i, Y_i, Z_i)$. While it would be possible to consider the sampling distribution over $(Y_i, Z_i)$ for a fixed vector of observed covariates $x_i$, doing so would make cross comparison of different stratifications impossible, because the sampling distribution would be over-specified relative to the coarser stratification. Because the PATE is of wide applied interest, we argue that averaging over observed control variables $X_i$ is sensible and all of our results are derived in this setting.

Another average treatment effect of broad interest is the {\em conditional average treatment effect} (CATE),
which defines an average treatment effect conditional on a set of covariate values. 
The population CATE,
$$\E(\tau(X) + \delta(X, \epsilon) \mid X = x) = \E(\tau(X) \mid X = x),$$
takes an expectation with respect to a conditional sampling distribution $\tau(X) \mid X = x$, where $\left\{X = x\right\}$ may denote a set of covariates rather than a 
single value. While the focus of this paper is on the PATE, its insights extend automatically to the population CATE.

The CATE is sometimes mistakenly reported in the literature as the {\em individual treatment effect} (ITE), which is a separate estimand that is only identified with more restrictive assumptions. 
The ITE is defined at the unit level as the difference in potential outcomes. For unit $i$, the ITE is given by
$$F(X_i, Z_i = 1, \epsilon_{i,y,1}) - F(X_i, Z_i = 0, \epsilon_{i,y,0}).$$
This is unidentified without further assumptions on the nature of the error term, as in general $\epsilon_{i,y,1} \neq \epsilon_{i,y,0}$; see Figure \ref{errors}.

\section{Minimal and optimal statistical control}

\subsection{The principal deconfounding function}

Although conditional unconfoundedness is central to our conception of causal effect estimation, in fact it is a stronger than necessary assumption for identifying the ATE. More specifically, one only needs a function $s(x)$ that satisfies {\em mean conditional unconfoundedness}.

\begin{definition}
A function $s$ on covariate space $\mathcal{X}$ is said to satisfy {\em mean conditional unconfoundedness} if
\begin{equation}
Z \independent  (\mu(X), \tau(X))  \mid s(X).
\end{equation}
\end{definition}

\begin{lemma}\label{MCU}
Mean conditional unconfoundedness is a sufficient condition for estimating average treatment effects.
\end{lemma}

\begin{proof}
Denote the causal model as $$Y^z \leftarrow \mu(X) + \upsilon(X,\epsilon_y) + (\tau(X) + \delta(X,\epsilon_y)) z$$ where $\epsilon_y \independent (Z, X)$, $\E(\upsilon(x,\epsilon_y))  = 0$, and $\E( \delta(x,\epsilon_y)) = 0$ for all $x$. We aim to show that $$\E(Y^z \mid s(X) = s) = \E(Y \mid s(X) = s, Z = z),$$ from which the result follows by the estimability of the right hand side for both $z = 0$ and $z = 1$.  Recalling the relationship between $Y^z$ and $Y \mid Z = z$ described in Section \ref{fcm},  this is equivalent to showing that 
\begin{align*}
\E(\mu(X) + &\upsilon(X,\epsilon_y) + (\tau(X) + \delta(X,\epsilon_y)) z \mid s(X) = s)  =\\
& \E(\mu(X) + \upsilon(X,\epsilon_y) + (\tau(X) + \delta(X,\epsilon_y)) z \mid s(X) = s, Z  = z),
\end{align*} 
where the expectation over $(X, \epsilon_y)$ is with respect to its marginal distribution on the left hand side and with respect to its conditional distribution, given $Z = z$, on the right hand side. By the independence of $\epsilon_y$, the mean zero errors for each $x$, and the linearity of expectation, this reduces to showing that $$\E(\mu(X) + \tau(X)z  \mid s(X) = s) = \E(\mu(X) +  \tau(X)z  \mid s(X) = s, Z  = z).$$ By the assumption of mean conditional unconfoundedness, $Z \independent  (\mu(X), \tau(X))  \mid s(X)$, and the result follows.
\end{proof}

Mean conditional unconfoundedness can be used to define a {\em minimal} control function, but first we must recall the definition of the propensity score \citep{rosenbaum1983central}, which we will denote by $\pi(\cdot)$.
\begin{definition}
The {\em propensity score}, based on a vector of control variables $x$, is the conditional probability of receiving treatment:
\begin{equation}\label{propscore}
\pi(x) \equiv \Prob(\rvZ = 1 \mid \rvX = x).
\end{equation}
\end{definition}
\noindent It is common to interchangeably refer to the propensity {\em score}, which emphasizes a specific numerical value, $\pi(x)$, and the propensity {\em function}, which emphasizes the mapping, $\pi: \mathcal{X} \rightarrow (0, 1)$.

In turn, we have:
\begin{definition}
The {\em principal deconfounding function} is given by following conditional expectation:
$$\lambda(x) = \E(\pi(X) \mid \mu(X) = \mu(x), \tau(X) = \tau(x)).$$
\end{definition}

\begin{theorem} \label{theorem1}
The principal deconfounding function is the coarsest function satisfying mean conditional unconfoundedness. 
\end{theorem}

\begin{proof}

By iterated expectation, $Z \mid \mu(X), \tau(X)$ is a Bernoulli random variable with probability $\lambda(X)$, therefore $$\E(Z \mid \tau(X), \mu(X), \lambda(X)) = \E(Z \mid \lambda(X)),$$ which shows that $Z \independent \left ( \mu(X), \tau(X)\right ) \mid \lambda(X)$
because $Z$ is binary.

Furthermore, $|\lambda(\mathcal{X})|$ is minimal: it takes exactly as many values as there are unique conditional distributions of $Z \mid \mu(X), \tau(X)$. In more detail, suppose $s(x)$ is coarser than $\lambda(x)$ so that there exists $x_1$ and $x_2$ such that $s(x_1) = s(x_2)$ but $\lambda(x_1) \neq \lambda(x_2)$. But $\lambda(x_1) \neq \lambda(x_2)$ implies $(\mu(x_1), \tau(x_1)) \neq (\mu(x_2), \tau(x_2))$, which in turn shows that $$Z \not \independent \mu(X), \tau(X) \mid s(X)$$ so mean conditional unconfoundedness is violated.
\end{proof}

\subsection{Optimal stratification for causal effect estimation} \label{trueprop}

Recognizing that valid control features are non-unique raises the question: which control features are the best ones? To make this question precise,  we study the finite sample variance of fixed-strata estimators, restricting our attention to a vector of discrete control variables. 

Without loss of generality, discrete control variables with finite support can be represented as a single covariate taking $K = |\mathcal{X}|$ distinct values. For example, a length $d$ vector of binary covariates would be represented as a single variable taking $2^d$ values. 
This assumption is mathematically convenient and, by setting $K$ large enough, can capture most empirical applications to a satisfactory degree of realism. (We revisit the plausibility of this assumption in the discussion section.) In the mathematical formalism and discussion of this paper, we will use the words ``strata" and ``features" interchangeably, to refer to functions of this single categorical variable.

In detail, this paper considers the following data generating process:
\begin{equation}\label{dgp_equation}
\begin{split}
\mathcal{X} &= \{1, \dots, K\},\\
\pi: \mathcal{X} &\mapsto (0, 1),\\
\rvZ &\sim \mbox{Bernoulli}(\pi(\rvX)),\\
\rvY &\leftarrow \mu(X) + \upsilon_X + (\tau(X) + \delta_X) Z
\end{split}
\end{equation}
where $\E(\upsilon_x) = 0$ and $\E(\delta_x) = 0$ for all $x$ so that $\mu(\scalarobs{x}) = \E(\rvY \mid \rvX = \scalarobs{x}, \rvZ = 0)$ and $\mu(\scalarobs{x}) + \tau(\scalarobs{x}) =  \E(\rvY \mid \rvX = \scalarobs{x}, \rvZ = 1)$. Lastly, let the random variable $\rvN$ denote the overall sample size and define subset-specific sample sizes as follows:
\begin{itemize}
\item $\rvN_{x}$: the number of observations with $\rvX = x$,
\item $\rvN_{x, z}$: the number of observations with $\rvX = x$ and $\rvZ = z$.
\end{itemize}
We define the stratification estimator using a stratification function $s\left(\mathcal{X}\right)$, which returns $J \leq K$ discrete function values. We compute the average difference 
in outcomes between the treated and control groups separately for individuals in each of the $J$ strata, so that
\begin{equation*}
\begin{aligned}
\bar{\tau}^{s}_{strat} &= \sum_{j \in s(\mathcal{X})} \frac{N_{j}}{n} \left( \bar{Y}_{j,1} - \bar{Y}_{j, 0} \right)\\
N_{j, 0} &= \sum_{i=1}^n \mathbf{1}\left\{s(X_i) = j\right\} \mathbf{1}\left\{Z_i = 0\right\}\\
\bar{Y}_{j,0} &= \frac{1}{N_{j, 0}} \sum_{i=1}^n Y_i \mathbf{1}\left\{s(X_i) = j\right\} \mathbf{1}\left\{Z_i = 0\right\}
\end{aligned}\;\;\;\;\;\;
\begin{aligned}
N_{j} &= \sum_{i=1}^n \mathbf{1}\left\{s(X_i) = j\right\}\\
N_{j, 1} &= \sum_{i=1}^n \mathbf{1}\left\{s(X_i) = j\right\} \mathbf{1}\left\{Z_i = 1\right\}\\
\bar{Y}_{j,1} &= \frac{1}{N_{j, 1}} \sum_{i=1}^n Y_i \mathbf{1}\left\{s(X_i) = j\right\} \mathbf{1}\left\{Z_i = 1\right\}\\
\end{aligned}
\end{equation*}
Note that if we choose the trivial stratification $s(x) = x$, we stratify completely on all $K$ unique levels of $\mathcal{X}$.

The following theorem describes when stratification beyond the minimal valid stratification, $\lambda(X)$, is beneficial, in terms of conditions on the underlying data generating process. 

\begin{theorem} \label{theorem2}
Assume we have stratified on $\lambda(X)$ so that the average treatment effect is identified 
using a minimal deconfounding set. Consider a {\em refinement} of $\lambda$, $s(X)$, which also identifies the ATE: 
$s(x) \neq s(x')$ while $\lambda(x) = \lambda(x')$ for at least two $x, x' \in \mathcal{X}$.
Define $\bar{\tau}_{\textrm{strat}}^{\lambda}$ as a stratification estimator which uses 
level sets of $\lambda(X)$ to define strata and $\bar{\tau}_{\textrm{strat}}^{s}$ as a 
stratification estimator which uses level sets of $s(X)$. 
Then $\V \left( \bar{\tau}_{\textrm{strat}}^{s} \right) < \V \left( \bar{\tau}_{\textrm{strat}}^{\lambda} \right)$ if $\nu < \eta$
where 
\begin{equation*}
\begin{aligned}
m(j) &= \lvert \left\{ s(x) : x \in \mathcal{X}\mbox{ such that } \lambda(x) = j \right\} \rvert\\
\mathcal{B} &= \left\{j \in \lambda(\mathcal{X}): m(j) > 1 \textrm{ and all sub-strata means and variances are constant} \right\}\\
\mathcal{C} &= \left\{j \in \lambda(\mathcal{X}): m(j) > 1 \textrm{ and either the sub-strata means or variances are non-constant} \right\}\\
\nu &= \sum_{b \in \mathcal{B}} \left[ \V\left( \frac{N_{b}}{n} \left( \bar{Y}_{b,1} - \bar{Y}_{b, 0} \right) \right) - \V\left( \sum_{\ell=1}^{m(b)}  \frac{N_{b\ell}}{n} \left( \bar{Y}_{b\ell,1} - \bar{Y}_{b\ell, 0} \right) \right)\right]\\
\eta &= \sum_{c \in \mathcal{C}} \left[ \V\left( \sum_{\ell=1}^{m(c)}  \frac{N_{c\ell}}{n} \left( \bar{Y}_{c\ell,1} - \bar{Y}_{c\ell, 0} \right) \right) - \V\left( \frac{N_{c}}{n} \left( \bar{Y}_{c,1} - \bar{Y}_{c, 0} \right) \right)\right]\\
\end{aligned}
\end{equation*}
and $\V \left( \bar{\tau}_{\textrm{strat}}^{s} \right) \geq \V \left( \bar{\tau}_{\textrm{strat}}^{\lambda} \right)$ otherwise.
\end{theorem}
A detailed proof is provided in Appendix \ref{appA}, but here we offer a sketch of the proof to build intuition.
In comparing two stratifications, $\lambda$ and $s$, across discrete covariates $X$, we can partition the 
level sets of the two stratfication functions as follows:
\begin{enumerate}
\item $\mathcal{A}$: values of $x \in \mathcal{X}$ for which both $\lambda$ and $s$ agree
\item $\mathcal{B}$: values of $x \in \mathcal{X}$ for which $s$ substratifies $\lambda$ but the mean and variance of $Y \mid Z$ are constant across substrata formed by $s$
\item $\mathcal{C}$: values of $x \in \mathcal{X}$ for which $s$ substratifies $\lambda$ and either the mean of $Y \mid Z$, the variance of $Y \mid Z$, or both vary across substrata formed by $s$
\end{enumerate}
We ignore $\mathcal{A}$ and focus on $\mathcal{B}$ and $\mathcal{C}$. In the case of $\mathcal{B}$, 
$s$ performs ``unnecessary" stratification, estimating and re-aggregating conditional means which are the same in the 
underlying data generating process, and thus incurs additional variance over the $\lambda$ stratification estimator. 
On the other hand, when we consider $\mathcal{C}$, $\lambda$ incurs additional variance over $s$ by failing to 
control for differences in the $Y \mid Z$. 

In summary, $\mathcal{B}$ induces a variance penalty on $s$ relative to $\lambda$ by ``overstratification", 
while $\mathcal{C}$ induces a variance penalty on $\lambda$ relative to $s$ by ``understratification." Which 
estimator is preferred depends on the magnitude of these competing effects, as articulated in the 
$\nu < \eta$ inequality above. 
The practical upshot of this theorem is that stratification that accounts for substantial variation in the response will tend to reduce variance of the treatment effect estimator (whether or not it is confounded in the sense of covarying with propensity to receive treatment), while stratification that accounts only for variation in treatment assignment will increase variance of the treatment effect estimator. This conclusion is illustrated in the examples of the following section.

\section{Vignettes}\label{examples}

This section collects examples illustrating the statistical trade-offs underlying feature selection for causal effect estimation that are articulated in Theorem \ref{theorem2}. Many of the examples are interesting in their own right; connections to previous literature are provided throughout. 

\subsection{In what sense is randomization the ``gold standard'' for causal effect estimation?}
It has become boiler-plate in reports on observational studies to remark that ``in the absence of the gold standard of a randomized clinical trial, one may pursue statistical methods to control for confounding''. But in what sense is randomized treatment assignment the gold standard? Surely solid-state physicists do not randomize their lab conditions and hope their sample size is large enough to reveal interesting results. Famously, esteemed physicist Ernst Rutherford quipped ``If your experiment needs statistics, you ought to have done a better experiment'' (\cite{hammersley1962monte}). The intuition behind this remark is that it is {\em control} that is central, not randomization. See section \ref{constantcontrol} for a definition of a control feature that evokes the experimental notion of ``control''. 

 Indeed, randomization is simply a way to guarantee control {\em on average} in the event that exact control is impossible, such as when crucial confounding factors are unobserved. This perspective in turn suggests that controlling for factors that we {\em can} observe and randomizing only for factors that we cannot observe would be the ideal approach.  The following thought experiment amplifies this intuition.

Consider studying the effect of treatment $Z$ on outcome $Y$ in a sample of 
$n$ pairs of identical twins and deciding how to allocate treatment across the $2n$ study participants. 
Completely randomized treatment assignment satisfies the assumptions outlined above and thus identifies the treatment effect. However, a naive randomization would sometimes accidentally treat both twins and leave other twin pairs untreated. This violates most people's intuition about why twin studies are interesting and useful, which is that giving one twin the treatment and the other a placebo implicitly ``controls for'' all of the shared biological and environmental factors that may impact the treatment effect. Randomization within each twin pair can protect against unmeasured factors that may confound the result, such as (perhaps) which twin was born first.

In this case, both $Z$ and the twin pair index, $X$, are informative about the 
expected value of $Y$. Now consider four possible approaches to study the effect of $Z$ on $Y$:
\begin{center}
\begin{tabular}{c | c | c} 
  & Design & Estimator \\ 
 \hline
 1 & Complete randomization & Unadjusted mean difference \\ 
 2 & Twin pair randomization & Unadjusted mean difference \\
 3 & Complete randomization & \;\;Adjusted mean difference \\ 
 4 & Twin pair randomization & \;\;Adjusted mean difference \\
\end{tabular}
\end{center}
where the unadjusted mean difference estimator is defined as 
$$\bar{\tau}_U = \bar{Y}_{Z=1} - \bar{Y}_{Z=0}$$
and the adjusted mean difference estimator is defined as 
$$\bar{\tau}_A = \sum_{x \in \mathcal{X}} \frac{n_x}{n} \left( \bar{Y}_{X=x, Z=1} - \bar{Y}_{X=x, Z=0} \right)$$
where $\mathcal{X}$ is the set of twin pairs and $X$ is a variable that indexes twin pairs.

Each of the four approaches above identifies the ATE. However, adjusting for twin pairs (approaches 3 and 4) will tend to reduce 
variance over the unadjusted alternatives (1 and 2) and, similarly, designs that incorporate twin pairs in randomization (2 and 4)
will also see a reduction in variance over the completely randomized alternatives (1 and 3). These results are implicit in Theorem \ref{theorem2}, which can be applied even if the propensity function is constant, as in a randomized trial. 

As intuitive as this example may be, and despite its lesson being a straightforward implication of Theorem \ref{theorem2}, regression adjustment for randomized trial data remains controversial. Freedman \citep{freedman2008regression, freedman2008randomization} criticized regression adjustment on the grounds that linear or linear logistic regression is potentially biased. Unfortunately, many researchers took this advice without first considering non-linear alternatives. \cite{lin2013agnostic} shows that regression adjustment in experimental data is not asymptotically 
unbiased if one entertains a richer set of interacted or saturated models, rather than a basic linear model. Of course, the stratification estimators studied here are fundamentally nonparametric and so are consistent with the conclusions of \cite{lin2013agnostic}. At the same time, Theorem \ref{theorem2} concedes that for some data generating processes, undertaking a regression adjustment (via stratification) would simply produce unnecessary variability, specifically for data generating processes where the available control factors are not sufficiently predictive of the response. 
In many applied problems we find ourselves somewhere in between this case of mostly useless controls and the twin experiment situation of profoundly informative controls. 

\subsection{Propensity scores}\label{propensity}
Following the work of \cite{rosenbaum1983central}, the propensity score (expression \ref{propscore}) has become a central element in many applied analyses of causal effects. In that paper, it was first shown that $\pi(x)$ satisfies conditional unconfoundedness, from which it follows that
\begin{equation}
\textrm{ATE} = \E[\rvY^1 - \rvY^0] = \E_{\pi(\rvX)}[\E[\rvY \mid \pi(\rvX), \rvZ = 1] - \E[\rvY \mid \pi(\rvX), \rvZ = 0]].
\end{equation}
This differs from the more general form of conditional unconfoundedness in that $\pi(\rvX)$ is one-dimensional, while $\rvX$ itself typically involves many controls. 

An especially common use of the propensity score in practice is via the inverse-propensity weighted (IPW) estimator
\begin{equation} \label{ipw}
\bar{\tau}_{\textrm{ipw}} = \frac{1}{\rvN} \sum_{i=1}^{\rvN} \left( \frac{\rvY_i \rvZ_i}{\pi(\rvX_i)} - 
\frac{\rvY_i (1-\rvZ_i)}{1-\pi(\rvX_i)} \right),
\end{equation}
which is known to be consistent and has been widely studied theoretically. 

Here we re-examine a curious result of \cite{hirano2003efficient} which shows that an IPW estimator based on estimated propensity scores attains lower asymptotic variance than one based on the true propensity function. We can apply the finite-sample results of Theorem \ref{theorem2} to re-evaluate the meaning of this widely-known result by noting the following correspondence between IPW estimators and stratification estimators:

\begin{lemma}\label{ipw_strat}
The empirical inverse propensity weighting (IPW) estimator is equivalent to $\bar{\tau}^{x}_{strat}$ under the following 
conditions:
\begin{enumerate}
\item $\mathcal{X}$ is discrete,
\item For all $x \in \mathcal{X}$, $N_{x,1} > 0$ and $N_{x,0} > 0$, 
\item The propensity weighting function is estimated nonparametrically as $\hat{\pi}(x) = N_{x, 1} / N_{x}$ for each $x \in \mathcal{X}$.
\end{enumerate}
\end{lemma}
\begin{proof}
By direct calculation,
\begin{equation*}
\begin{aligned}
\bar{\tau}^{x}_{ipw} &= \frac{1}{n} \sum_{i=1}^n \left(\frac{Y_i Z_i}{\hat{\pi}(X_i)} - \frac{Y_i(1-Z_i)}{1-\hat{\pi}(X_i)}\right) \\
&= \frac{1}{n} \sum_{i=1}^n \left(\frac{Y_i Z_i}{N_{x_i, 1} / N_{x_i}} - \frac{Y_i(1-Z_i)}{N_{x_i, 0} / N_{x_i}}\right) = \frac{1}{n} \sum_{i=1}^n \left(\frac{Y_i Z_i N_{x_i}}{N_{x_i, 1}} - \frac{Y_i (1-Z_i) N_{x_i}}{N_{x_i, 0}}\right)\\
&= \frac{1}{n} \sum_{x \in \mathcal{X}} \left( \frac{N_x}{N_{x,1}} \left( \sum_{i: X_i = x} Y_i Z_i \right) -  \frac{N_x}{N_{x,0}} \left( \sum_{i: X_i = x} Y_i (1 - Z_i) \right)\right)\\
&= \frac{1}{n} \sum_{x \in \mathcal{X}} \left( \frac{N_x}{N_{x,1}} \left( N_{x,1} \bar{Y}_{x,1} \right) -  \frac{N_x}{N_{x,0}} \left( N_{x,0} \bar{Y}_{x,0} \right)\right)\\
&= \sum_{x \in \mathcal{X}} \frac{N_x}{n} \left(\frac{N_{x,1} \bar{Y}_{x,1}}{N_{x,1}} - \frac{N_{x,0} \bar{Y}_{x,0}}{N_{x,0}}\right) = \sum_{x \in \mathcal{X}} \frac{N_x}{n} \left(\bar{Y}_{x,1} - \bar{Y}_{x,0}\right) = \bar{\tau}^{x}_{strat}.
\end{aligned}
\end{equation*}
\end{proof}

First, we give a finite-sample analogue of the \cite{hirano2003efficient} result in the stratification context. Then, we demonstrate a modified estimator based on a known propensity score that improves upon the estimated propensity score IPW.

\subsubsection{``Noisy estimates of one''.}
Denote a candidate propensity function by $q: \mathcal{X} \mapsto (0, 1)$, so that the corresponding IPW estimator is
\begin{equation}
\bar{\tau}_{\textrm{ipw}}^q = \sum_x \left( \frac{\rvN_x}{\rvN} \right) \bar{\tau}_{\textrm{ipw}}^{q,x}
\end{equation}
where
\begin{equation}
\bar{\tau}_{\textrm{ipw}}^{q,x} = \left( \frac{\hat{\pi}(x)}{q(x)} \bar{\rvY}_{x, \rvZ=1} - \frac{1-\hat{\pi}(x)}{1-q(x)} \bar{\rvY}_{x, \rvZ=0} \right)
\end{equation}
and  $\hat{\pi}(x) = \left( \rvN_{x, \obsZ=1} / \rvN_{x} \right)$ is the proportion of treated units in each stratum.

Taking $q(x) = p(x)$ is the ``true propensity score'' case, while letting $q(x) = \hat{\pi}(x)$ is the ``estimated propensity score'' case. In the former case, the treated and untreated stratum averages are weighted by $\hat{\pi}(x) / \pi(x)$ and $\left( 1-\hat{\pi}(x) \right) / \left(1-\pi(x) \right)$, respectively; in the latter case the weights are identically one. This difference in weights leads to the following analogue of the result of \cite{hirano2003efficient}:

\begin{theorem} \label{theorem3}
There exists some $\epsilon > 0$ such that if $\lvert \mu(x) \rvert + \lvert \tau(x) \rvert > \epsilon$ for at least one $x \in \mathcal{X}$, 
$$\V \left( \sum_{x \in \mathcal{X}} \bar{\tau}_{\textrm{ipw}}^{\hat{\pi},x} \right) \leq \V \left( \sum_{x \in \mathcal{X}} \bar{\tau}_{\textrm{ipw}}^{\pi,x} \right).$$
\end{theorem}

Essentially, the random weights in the true propensity IPW are only adding variability, compared to the IPW based on estimated weights, where exact cancellation occurs. A proof may be found in the appendix. Of course, there are many other possible IPW estimators, such as those based on parametric estimates. However, any parametric form will have a similar problem to the true propensity IPW if exact cancellation is not obtained. 

\subsubsection{The dimension reduction benefit of known propensity scores.}
Armed with an understanding of why the estimated propensity weights outperform the true propensity weights permits us to consider a modified estimator that is able to make use of knowledge of the true propensity scores (should they be known). Suppose $K_{\pi} = |\pi(\mathcal{X})| < |\mathcal{X}| = K$. If $\pi$ were known exactly prior to estimating the average treatment effect, this reduction in the strata should confer a benefit in terms of variance reduction --- there are simply fewer conditional expectations to estimate and there are more data available for estimating each one. Moreover, it is still possible to avoid the noisy-estimation-of-one effect by estimating the propensity score values on the level sets of $\pi$; letting $\rho \in \pi(\mathcal{X})$ denote a specific value in the range of $\pi$ gives

\begin{equation}
\begin{split}
\bar{\tau}_{\textrm{ipw}}^{\hat{\pi},\rho} &= \left( \frac{\hat{\pi}(\rho)}{\hat{\pi}(\rho)} \bar{\rvY}_{\rho, \rvZ=1} - 
\frac{1-\hat{\pi}(j)}{1-\hat{\pi}(\rho)} \bar{\rvY}_{\rho, \rvZ=0} \right)\\
&=  \bar{\rvY}_{\rho, \rvZ=1} - 
 \bar{\rvY}_{\rho, \rvZ=0} .
\end{split}
\end{equation}
More precisely:

\begin{corollary} \label{corollary1}
Suppose the following conditions hold:
\begin{enumerate}
\item If $\pi(x) = \pi(x')$, then $\mu(x) = \mu(x')$ and $\tau(x) = \tau(x')$,
\item $\V(\epsilon \mid \rvX = x) = \sigma_{j}^2$ for all $x$ with $\pi(x) = j$ and for all $j \in \pi(\mathcal{X})$, and
\item $|\pi(\mathcal{X})| <  |\mathcal{X}|$.
\end{enumerate}
Then 
$$\V \left( \sum_{j \in \pi(\mathcal{X})} \bar{\tau}_{\textrm{ipw}}^{\hat{\pi},j} \right) \leq \V \left( \sum_{j \in \pi(\mathcal{X})} \left( \sum_{x: \pi(x) = j} \frac{N_x}{N_j} \bar{\tau}_{\textrm{ipw}}^{\hat{\pi},x} \right) \right).$$
\end{corollary}
This result formalizes the intuition that fewer strata implies a greater degree of aggregation and that, with larger sample sizes in the remaining strata, estimation should be accordingly more efficient. In other words, knowledge of the true propensity function permits feature selection, after which the empirical propensities can be used in an IPW estimator (which is equivalent to the stratification estimator on the selected features).

Condition one requires some explanation: in the fixed-strata case ``over-stratification'' can actually be beneficial if the additional strata are predictive of the response itself and condition one rules out this possibility, as it states that $\mu\mbox{-}\tau$ is at least as coarse as $\pi$. That is, in addition to the ``noisy-estimation-of one'' phenomenon, empirical estimates of the propensity score can benefit from being defined on strata that are predictive of the response, but {\em not} the treatment assignment; this benefit is not directly related to the true-versus-actual propensity score question, but merely reflects the fact that controlling for prognostic factors can benefit treatment effect estimation. 

\subsubsection{The inefficiency of instrumental controls.}
While a known propensity score can potentially aid IPW estimation by preventing unnecessary stratification, an additional corollary of Theorem \ref{theorem2} tells us that stratification based on a known propensity function may produce unnecessary stratification as a result of {\em unconfounded} variation in propensity scores, which we refer to as ``instrumental'' stratification.

\begin{corollary} \label{corollary2}
Define a stratification $s$ such that $\lvert s(\mathcal{X}) \rvert < \lvert \pi(\mathcal{X}) \rvert$ and define 
$g: \pi(\mathcal{X}) \rightarrow s(\mathcal{X})$ as a function that collapses level sets of $\pi$ into level sets of $s$. 
Let $m(j) = \lvert \left\{ \pi(x): g(\pi(x)) = j \right\} \rvert$ and suppose the following conditions hold:
\begin{enumerate}
\item There exist $x, x'$ such that $\pi(x) \neq \pi(x')$ while $s(x) = s(x')$, $\mu(x) = \mu(x')$, and $\tau(x) = \tau(x')$,
\item $\V(\epsilon \mid \pi(\rvX) = p) = \sigma_{j}^2$ for all $x$ with $\pi(x) = p$ and $g(\pi(x)) = j$ and for all $j \in s(\mathcal{X})$
\end{enumerate}
Then, 
$$\V \left( \sum_{j \in s(\mathcal{X})} \bar{\tau}_{\textrm{ipw}}^{\hat{\pi},j} \right) \leq \V \left( \sum_{j \in s(\mathcal{X})} \left( \sum_{\pi: g(\pi) = j} \frac{N_{\pi}}{N_j} \bar{\tau}_{\textrm{ipw}}^{\hat{\pi},\pi} \right) \right).$$
\end{corollary}
This corollary and the other results of this subsection provide rigorous finite-sample corroboration of the advice offered in \cite{hernan2020causal} quoted in the introduction.

\subsection{Generalized prognostic scores}
In data generating processes where variation in $\tau$ is independent of $Z$, the {\em prognostic score}, $\E(Y^0 \mid X = x) = \mu(x)$, is a sufficient control function. This follows because mean conditional unconfoundedness is satisfied trivially by $s(X) = \mu(X)$ when $\tau(X) \independent Z$; see Lemma \ref{MCU}.  Like the propensity score, the prognostic score can be estimated from partially observed data --- the propensity score can be estimated from $(X, Z)$ pairs and the prognostic score can be estimated from control units only, $(X, Z = 0, Y)$, which in many contexts are more readily available than treated observations. See \cite{hansen2008prognostic} for a rigorous exposition of prognostic scores. 

The vector-valued function $(\mu, \tau)$ is a ``generalized'' prognostic score, containing both the usual prognostic score, as well as the treatment effect itself. This version of the prognostic score has received little attention, presumably because it ``begs the question'', in that one of its elements is the very estimand of interest. However, note that conditioning on a random variable is not about the values of that variable per se, but is rather about the level sets of the function defining that random variable. In particular, any one-to-one function of $\mu$-$\tau$ also satisfies mean conditional unconfoundedness; knowledge of the treatment effect itself is not required, merely knowledge of which strata have distinct treatment effects. Note also that Theorem \ref{theorem2} suggests that prognostic strata are more desirable from an estimation variance perspective, suggesting, perhaps counterintuitively, that large control groups may be advantageous in practice and that investing in data collection of prognostic factors should be prioritized in cases where randomization of treatment assignment is not possible.

\subsection{Constant control function}\label{constantcontrol}

The previous two examples showed that propensity scores and prognostic scores are sufficient control functions; this example demonstrates a function that may be coarser than either one. Consider a function on $\mathcal{X}$ defined as follows:
\begin{definition}
A function $s$ on $\mathcal{X}$ is a {\em constant control} function if for all $x, x' \in \mathcal{X}$ such that $s(x) = s(x')$ at least one of the following holds
\begin{itemize}
\item $\pi(x) = \pi(x')$,
\item $\mu(x) = \mu(x')$ and $\tau(x) = \tau(x')$.
\end{itemize}
\end{definition}
In other words, a constant control function is a coarsening of $\mathcal{X}$ such that on each level set defined by $s$, either $\pi(x)$ or $(\mu(x), \tau(x))$ are constant. The following lemma shows that a constant control function defines a random variable $S = s(X)$ such that $\E(Y \mid Z = z, S) = \E(Y^z \mid S)$. 

\begin{lemma} \label{lemma4}
Assume $X$ satisfies conditional unconfoundedness and consider the random variable $S = s(X)$, where $s$ is a constant control function; then $S$ satisfies conditional unconfoundedness. 
\end{lemma}

\begin{figure}
\begin{tikzpicture}
	\begin{pgfonlayer}{nodelayer}
		\node [style=myvar] (0) at (1.5, 5.5) {$X$};
		\node [style=myvar] (1) at (-2.25, 2.5) {$X_p$};
		\node [style=myvar] (2) at (1.5, 2.5) {$X_c$};
		\node [style=myvar] (3) at (5, 2.5) {$X_y$};
		\node [style=myvar] (5) at (-2.25, -1.25) {$Z$};
		\node [style=myvar] (6) at (5, -1.25) {$Y$};
		\node [style=none] (7) at (1.5, -2.25) {(a)};
	\end{pgfonlayer}
	\begin{pgfonlayer}{edgelayer}
		\draw [style=arrow] (1) to (5);
		\draw [style=arrow, in=90, out=-90] (3) to (6);
		\draw [style=arrow] (0) to (1);
		\draw [style=arrow] (0) to (3);
		\draw [style=arrow] (0) to (2);
		\draw [style=arrow] (5) to (6);
	\end{pgfonlayer}
\end{tikzpicture}\begin{tikzpicture}
	\begin{pgfonlayer}{nodelayer}
		\node [style=myvar] (1) at (-2.75, 2.5) {$X_p$};
		\node [style=myvar] (2) at (1.5, 2.5) {$X_c$};
		\node [style=myvar] (3) at (5.5, 2.5) {$X_y$};
		\node [style=myvar] (5) at (-2.75, -1.25) {$Z$};
		\node [style=myvar] (6) at (5.5, -1.25) {$Y$};
		\node [style=none] (8) at (1.75, -2.25) {(b)};
	\end{pgfonlayer}
	\begin{pgfonlayer}{edgelayer}
		\draw [style=arrow] (1) to (5);
		\draw [style=arrow, in=90, out=-90] (3) to (6);
		\draw [style=arrow] (5) to (6);
		\draw [style=probedge, in=120, out=60, looseness=1.25] (1) to (3);
		\draw [style=probedge] (1) to (2);
		\draw [style=probedge] (2) to (3);
	\end{pgfonlayer}
\end{tikzpicture}
\caption{Causal graphical model and partially causal graphical model after integrating out $X$.}
\label{decomposed_graphs}
\end{figure}
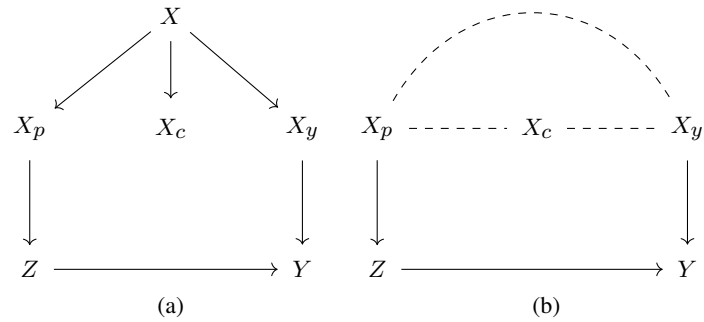

\begin{proof}
Consider the causal diagram of $(X, Z, Y)$ expanded to include random variables $X_p = \pi(X)$, $X_y = (\mu(X), \tau(X))$, and $X_c = s(X)$ for $s$ defined above, depicted in panel (a) of Figure \ref{decomposed_graphs}. Integrating out $X$ leads to a probabilistic graphical model as shown in panel (b) of Figure \ref{decomposed_graphs}; dashed lines denote not-necessarily causal probabilistic dependence and solid arrows denote causal relationships. The result follows by showing that $X_p \independent X_y \mid X_c$; in terms of the diagram this means that the curved dashed line does not exist. But this follows immediately from the definition of $X_c$. For any value of $X_c$, either $X_p$ or $X_y$ is constant, and so the conditional distribution of $X_p$ and $X_y$ is trivially a product distribution. 
\end{proof}
\noindent The intuition behind a constant control function is that one way to control for ``systematic co-variation'' is simply to remove all variation. Clearly, both $\pi(X)$ and $(\mu(X), \tau(X))$ are themselves constant control functions, as is $X$ itself. However, a constant control function may be coarser than either, as illustrated in Figure \ref{CCDR}, which shows an example of a simple data generating process that has a constant control function. In this example, the CCDR comprises just two strata, although $\mu$ and $\pi$ take 10 and 11 unique values, respectively, and $|\mathcal{X}| = 20$.  The treatment effect is heterogeneous but unconfounded: $\tau \sim \mbox{U}(5,10)$. The second panel of Figure \ref{CCDR} shows the sampling distributions of four different stratification estimators: one using level sets of $\mu$, one using level sets of $\pi$, one using all 20 values of $x$, and one using the two values of the minimal constant control function, indicating if $x \leq 11$. All four stratification estimators are unbiased, but their differing variances exhibit a pattern consistent with Theorem \ref{theorem2}: $\mu$ gives the lowest variance, followed by $x$, followed by the constant control function, followed by $\pi$.

\begin{figure}
\hspace*{-0.5cm} 
\includegraphics[width=3in]{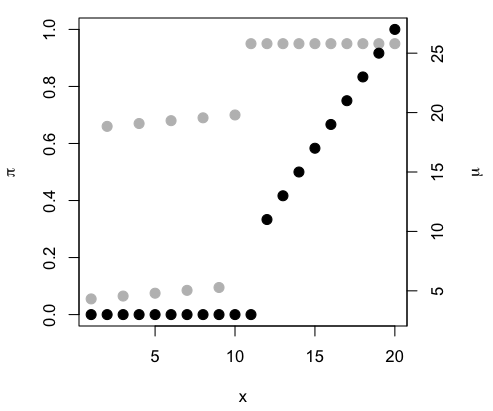}\includegraphics[width=3in]{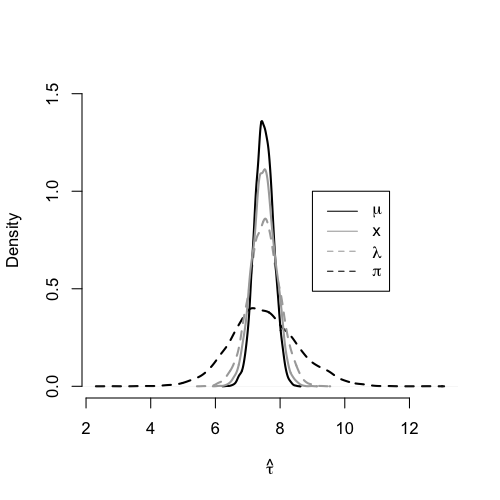}
\caption{An example of a DGP admitting a simple constant control function, $\lambda = \mathds{1}(x \leq 11)$. Here $\tau \sim \mbox{U}(5,10)$ is heterogeneous and $x \in \lbrace 1, \dots 20 \rbrace$. The left panel shows the $\mu$ values in black and the $\pi$ values in gray. The right panel shows the sampling distributions of stratification estimators based on the level sets of different function: $\mu$ (solid black), $x$ (solid gray), $\lambda$ (dashed gray) and $\pi$ (dashed black). All four estimators are unbiased, with variances that differ in line with the results of Theorem \ref{theorem2}.}
\label{CCDR}
\end{figure} 

\begin{figure}
 \begin{tikzpicture}
   \begin{scope}
    \fill[lightgray]  \thirdcircle;
    \fill[lightgray]  \secondcircle;
          \end{scope}
          \begin{scope}
          \fill[white] \firstcircle;
          \end{scope}
             \begin{scope}
                 \clip \thirdcircle;
    \fill[draw=black, pattern=north east lines]  \firstcircle;
          \end{scope}
                 \begin{scope}
                 \clip \secondcircle;
    \fill[draw=black, pattern=north east lines] \firstcircle;
          \end{scope}
      \draw \firstcircle node[text=black,above] {$\pi$};
      \draw \secondcircle node [text=black,below left] {$\mu$};
      \draw \thirdcircle node [text=black,below right] {$\tau$};
    \end{tikzpicture}
    \caption{If the coordinate dimensions of $\rvX$ are independent, then which variables $X_j$ appear (or not) in the eight possible combinations of $\pi$, $\mu$, and $\tau$ can be used to characterize four relevant variable types with respect to treatment effect estimation: necessary controls, pure prognostic variables, instruments, and extraneous (or noise) variables. The above Venn diagram depicts these eight regions, shaded according to these designations. Variables in the cross-hatched region are necessary controls, as they appear in both $\pi$ and either $\mu$ or $\tau$ (or both). The gray shaded region corresponds to pure prognostic variables, appearing in $\mu$ or $\tau$ (or both), but not appearing in $\pi$. The white region corresponds to instruments, variables which appear in $\pi$, but neither in $\mu$ nor $\tau$. Variables outside of the three circles are entirely irrelevant to either the outcome or the treatment. Such designations become considerably more complicated when the elements of $X$ are not independent (cf. Example \ref{noncausalSEM}).}
\label{venn}
\end{figure}
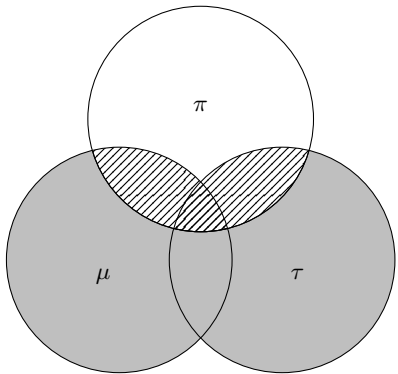

\subsection{Independent variables in both $\pi$ and $(\mu, \tau)$.}\label{independent}
When the coordinates of $\rvX$ (the nodes in the graph) are all mutually independent, a valid control set is the elements $X_j$ occurring in {\em both} the propensity model and (at least one of) the prognostic  and moderation models. For example, this was the strategy used in concocting the example DGP presented in section \ref{partial}. As a more general example, if $\pi(x_1, \dots, x_d) = \pi(x_1, x_2)$, $\mu(x_1, \dots, x_d) = \mu(x_2, x_3)$, and $\tau(x_1, \dots, x_d) = \tau(x_4, x_5)$, and $X_1 \independent X_{j}$ for $j \neq 1$, then $X_2$ is a sufficient control. This is because $X_1$ can be integrated out of the model without inducing dependence between $\pi(x_2) = \E(\pi(X_1, x_2) \mid X_2 = x_2)$ and $(\mu, \tau)$, because $X_1$ is independent of variables appearing in $\mu$ and/or $\tau$ (and does not itself appear). A similar integration could be performed for variables in $\mu$ and/or $\tau$ that do not appear in $\pi$, so long as it too was independent. In fact, only conditional independence is necessary; in the present example, $X_1 \independent (X_3, X_4, X_5) \mid X_2$. Figure \ref{venn} depicts the characterization of variables into four categories (necessary controls, pure prognostics variables, instruments, and extraneous) in the case that they are mutually independent. 

\subsection{Sets satisfying the back-door criterion according to a given CDAG.}\label{noncausalSEM}
Consider the causal diagram in Figure \ref{rectangle}. Either the propensity controls $(X_1, X_3)$ or the prognostic-moderation controls $(X_2, X_4)$ are adequate for statistical control. However, Pearl's algorithm tells us that ``mixed'' variables also suffice, such as $(X_1, X_4)$ or $(X_2, X_3)$. Interestingly, such examples show that the notion of ``instrumental'' variables and ``prognostic'' variables are context dependent. Specifically, relative to a conditioning set of $(X_2, X_3)$, additional stratification using $X_4$ is prognostic, while additional stratification on $X_1$ would be instrumental. Theorem \ref{theorem2} suggests that adding prognostic controls is often desirable, while adding instruments should be avoided, but such designations will fluctuate depending on what has already been included.

This example also illustrates a limitation of the triangle graph. Suppose that only $(X_2, X_3)$ were available for measurement. The resulting diagram for just those two controls (Figure \ref{triangle}) is {\em not} the usual causal diagram, because $X_2$ has no causal impact on $Z$, while $X_3$ has no causal impact on $Y.$ Accordingly, there is no unaugmented CDAG describing $(X_2, X_3, Z, Y)$; instead, we must denote merely statistical relationships using dashed lines. When a practitioner invokes conditional unconfoundedness in the potential outcomes framework, it therefore does not imply the triangular CDAG. 

Similarly, invoking (conditionally) exogenous errors does not imply that the resulting mean components of the structural model are causal. In more detail, if the potential outcomes are defined in terms of the CDAG on the full set $(X_1, X_2, X_3, X_4)$, a structural model can be derived that only involves $(X_2, X_3)$, as follows:

\begin{equation}
\begin{split}
Y^0 &=  F(x_1, x_2, x_3, x_4, z = 0, \epsilon_y) = F(x_2, x_3, z = 0, \epsilon_y)\\
Y^1 &= F(x_1, x_2, x_3, x_4, z = 1, \epsilon_y)  = F(x_2, x_3, z = 1, \epsilon_y)\\
\mu(x_2, x_3) &\equiv \E(Y^0 \mid X_2, = x_2, X_3 = x_3)\\
\tau(x_2, x_3) &\equiv \E(Y^1 \mid X_2 = x_2, X_3 = x_3) -  \E(Y^0 \mid X_2, = x_2, X_3 = x_3)\\
\upsilon(X_1, x_2, x_3, X_4, \epsilon_y) &\equiv F(X_1, x_2, x_3, X_4, z = 0, \epsilon_y) - \mu(x_2, x_3)\\
& = F(x_2, X_4, z = 0, \epsilon_y) - \mu(x_2, x_3)\\
\delta(X_1, x_2, x_3, X_4, \epsilon_y) &\equiv F(X_1, x_2, x_3, X_4, z = 1, \epsilon_y)  - F(X_1, x_2, x_3, X_4, z = 0, \epsilon_y) - \tau(x_2, x_3)\\
&= F(x_2, X_4, z = 1, \epsilon_y)  - F(x_2, X_4, z = 0, \epsilon_y) - \tau(x_2, x_3).
\end{split}
\end{equation}
Noting that the resulting error terms now depend not only $\epsilon_y$, but also on $X_4$, it is necessary to show that $$(X_4, \epsilon_y) \independent Z \mid (X_2, X_3).$$ But this follows from the fact that $\E(Z \mid X_2 = x_2, X_3 = x_3) = \E(\pi(X_1, X_3) \mid X_2 = x_2, X_3 = x_3) \equiv \pi(x_2, x_3)$ is free of $X_4$. In this model, $\mu(x_2, x_3)$, $\tau(x_2, x_3)$ and $\pi(x_2, x_3)$ must not be interpreted as causal functions, despite yielding the required exogenous errors; specifically, from the graph we know that $X_2$ has no causal impact on $Z$ and $X_3$ has no causal impact on $Y,$ as depicted in Figures \ref{rectangle} and \ref{triangle}.

\subsection{Sets satisfying the back-door criterion according to a transformed CDAG.} \label{transformedCDAG}

This example considers a data generating process that admits distinct CDAGs, depending on how the control variables are parametrized. This scenario is not commonly discussed, presumably because observed measurements are taken to be designated by ``nature'', so to speak. However, reflecting on invertible transformations such as $(x_1, x_2) \rightarrow (x_1, x_1/x_2)$ highlights that functional causal models are certainly subject to changes of variables.

More concretely, consider the following DGP: 
\begin{equation*}
\begin{aligned}
X_j &\stackrel{iid}{\sim}\mbox{Bernoulli}\left(p_j\right)\\
\pi(X) &= \beta_0 + \beta_1 (2X_1X_2 - X_1 - X_2 + 1) + \beta_2 X_3\\
Z &\sim \mbox{Bernoulli}\left(\pi(X)\right)\\
\mu(X) &= \alpha_0 + \alpha_1 (2X_1X_2 - X_1 - X_2 + 1) + \alpha_2 X_4\\
\tau(X) &= \tau \;\; \mbox{(constant treatment effect)}\\
Y &= \mu(X) + \tau(X) Z + \epsilon, \;\;\;  \epsilon \sim \mathcal{N}\left(0, \sigma^2_{\epsilon}\right)\\
\end{aligned}
\end{equation*}

Next, define random variable $W = (2X_1X_2 - X_1 - X_2 + 1)$, regarding $X_2$ as the 
exogenous variable in the functional model for $W \mid X_1$. Additionally,  suppressing 
$X_3$ and $X_4$, as they represent exogenous variation, yields the 
causal graph in Figure \ref{graph5}.

\begin{figure}
\begin{minipage}[b]{180pt}
\centering
\begin{tikzpicture}[baseline=-0.25em,scale=0.45]
	\begin{pgfonlayer}{nodelayer}
		\node [style=myvar] (1) at (2.5, 2.5) {$X_1$};
		\node [style=myvar] (2) at (2.5, -2.5) {$X_2$};
		\node [style=myvar] (3) at (-2.5, -2.5) {$X_3$};
		\node [style=myvar] (4) at (7.5, 2.5) {$X_4$};
		\node [style=myvar] (5) at (-2.5, 0) {$Z$};
		\node [style=myvar] (6) at (7.5, 0) {$Y$};
	\end{pgfonlayer}
	\begin{pgfonlayer}{edgelayer}
		\draw [style=arrow] (1) to (5);
		\draw [style=arrow] (1) to (6);
		\draw [style=arrow] (2) to (5);
		\draw [style=arrow] (2) to (6);
		\draw [style=arrow] (3) to (5);
		\draw [style=arrow] (4) to (6);
		\draw [style=arrow] (5) to (6);
	\end{pgfonlayer}
\end{tikzpicture}
\caption{Causal graph in terms of original covariates}
\label{graph5a}
\end{minipage}
\hfill
\begin{minipage}[b]{180pt}
\centering
\begin{tikzpicture}[baseline=-0.25em,scale=0.75]
	\begin{pgfonlayer}{nodelayer}
		\node [style=myvar] (1) at (-2.5, 2.5) {$X_1$};
		\node [style=myvar] (2) at (2.5, 2.5) {$W$};
		\node [style=myvar] (7) at (0.66, 0) {$Z$};
		\node [style=myvar] (8) at (4.33, 0) {$Y$};
	\end{pgfonlayer}
	\begin{pgfonlayer}{edgelayer}
		\draw [style=arrow] (1) to (2);
		\draw [style=arrow] (2) to (7);
		\draw [style=arrow] (2) to (8);
		\draw [style=arrow] (7) to (8);
	\end{pgfonlayer}
\end{tikzpicture}
\caption{Causal graph under transformed covariates}
\label{graph5}
\end{minipage}
\end{figure}

From this graph, it is clear that conditioning on $W$ satisfies conditional unconfoundedness.
Most interestingly, $\lvert \mu(\mathcal{X}) \rvert = \lvert \pi(\mathcal{X}) \rvert = 4$.
while $\lvert \mathcal{W} \rvert = 2$; thus $W$ provides the smallest possible random variable. Indeed, the level sets of $W$ are exactly the level sets of $\lambda$: $$\E(\pi(X) \mid \mu(X)) = \E(\pi(X) \mid W, X_4) = \E(\pi(X) \mid W).$$

\subsection{Sets that induce collider bias in a graph without colliders}\label{pseudoCollider}

We see in Section \ref{transformedCDAG} that conditioning on synthetic ``features" that combine existing 
variables can lead to smaller control sets than their component variables. 
It is thus perhaps natural to consider machine learning useful in searching for and constructing such sets. 
It is true that \textit{certain} combinations of confounding variables may create a synthetic, 
minimal deconfounders. However, it is also possible to combine two independent variables to create a 
``collider" (defined in Section \ref{DAG_section}) which confounds the causal effect of $Z$ on $Y$ after 
conditioning.

Consider the graph in Figure \ref{pseudo_collider_graph} and define its data generating equations as 
\begin{equation*}
\begin{aligned}
Y &\sim \mathcal{N}\left(\alpha X_2 + \tau Z , \sigma^2 \right)\\
Z &\sim \mbox{Bernoulli}\left(1 / 4 + X_1 / 2\right)\\
X_1, X_2 &\sim \mbox{Bernoulli}\left(1 / 2\right)
\end{aligned}
\end{equation*}

\begin{figure}
\centering
\begin{tikzpicture}[baseline=-0.25em,scale=0.45]
	\begin{pgfonlayer}{nodelayer}
		\node [style=myvar] (1) at (-3, 3) {$X_1$};
		\node [style=myvar] (2) at (3, 3) {$X_2$};
		\node [style=myvar] (3) at (-3, 0) {$Z$};
		\node [style=myvar] (4) at (3, 0) {$Y$};
	\end{pgfonlayer}
	\begin{pgfonlayer}{edgelayer}
		\draw [style=arrow] (1) to (3);
		\draw [style=arrow] (2) to (4);
		\draw [style=arrow] (3) to (4);
	\end{pgfonlayer}
\end{tikzpicture}
\caption{Graph with no confounding and no colliders}
\label{pseudo_collider_graph}
\end{figure}
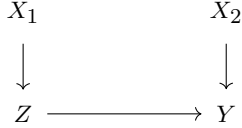

From this graph, we can see that the average causal effect of $Z$ on $Y$ is identified unconditional of $X_1$ and $X_2$, though 
we may condition on either or both variables. 
Suppose we construct two synthetic variables
\begin{equation*}
\begin{aligned}
\tilde{X}_A &= \min\left\{X_1, X_2\right\}\\
\tilde{X}_B &= a \left[ \mathbf{1}\left\{X_1 == 1\right\} \mathbf{1}\left\{X_2 == 1\right\} + \mathbf{1}\left\{X_1 == 0\right\}\mathbf{1}\left\{X_2 == 0\right\} \right]\\
&\;\;\;\;\; + b \left( \mathbf{1}\left\{X_1 == 1\right\} \mathbf{1}\left\{X_2 == 0\right\} \right) + c \left( \mathbf{1}\left\{X_1 == 0\right\} \mathbf{1}\left\{X_2 == 1\right\} \right)
\end{aligned}
\end{equation*}
where the unique values of categorical $\tilde{X}_B$ may be treated as strata of a conditioning set.
We show in the simulation results in Table \ref{tab:table2} that conditioning on either $\tilde{X}_A$ or $\tilde{X}_B$ biases the average treatment effect, 
while conditioning on both $X_1$ and $X_2$ does not. 
Note that both $\tilde{X}_A$ and $\tilde{X}_B$ are constructed in a manner not unlike the ``feature learning" step of common machine learning algorithms, 
such as neural networks and decision trees. 

\begin{table}\caption{\label{tab:table2}Simulation results for $10,000$ simulations of $n = 1,000$ of the above DGP}
\centering
\begin{tabular}{c|cc}
\hline
Control Set & Bias & Log-likelihood\\
\hline
$\varnothing$ & 0.00 & -3,143\\
$X_2$ & 0.00 & -2,352\\
$X_1, X_2$ & 0.00 & -2,350\\
$\tilde{X}_A$ & -1.72 & -2,967\\
$\tilde{X}_B$ & 2.50 & -2,798\\
\end{tabular}
\end{table}
%

\subsection{Sets satisfying the back-door criterion with respect to a mean CDAG.}\label{meanCDAG}
The structural model perspective permits us to produce, starting from a given CDAG, a modified causal diagram that reflects only the mean dependencies. For estimation of average causal differences, such a graph suffices to identify valid control variable sets that are potentially smaller than any control set satisfying the back-door criterion on the original CDAG. 

For example, consider the following data generating process:
\begin{equation*}
\begin{aligned}
X &\sim \mbox{Bernoulli}(1/2),\\
Z &\sim \mbox{Bernoulli}\left(1/4 + X/2 \right)\\
Y &\sim \mathcal{N}\left(\tau Z, (\sigma + X)^2\right)
\end{aligned}
\end{equation*}
For this DGP, $\mu(X) = 0$ and $\tau(X) = \tau$ are both constant in $X$, which implies that the null set satisfies mean conditional unconfoundedness; even though $X$ is a common cause of $Z$ and $Y$, it only affects the variance of $Y$, but not the mean. Therefore, the full joint distribution of $X, Z, Y$ is the triangle diagram of Figure \ref{triangle}, 
while Figure \ref{graph16} depicts the joint distribution of $(X, Z, \E (Y \mid X, Z))$, in which $X$ is unconnected to $\E(Y \mid X, Z) = \E(Y \mid Z)$.

Note that while mean conditional unconfoundedness identifies the ATE, it does not identify other causal estimands. For instance, consider the
quantile treatment effect (QTE), for $q \in (0,1)$:
$$F^{-1}_{Y^1}(q) - F^{-1}_{Y^0}(q)$$
where $F^{-1}$  denotes an inverse cumulative distribution function.
Integrating out $X$, $Y \mid Z = z$ is a mixture of two normal random variables, with PDF and CDF defined as 
\begin{equation*}
\begin{aligned}
f(y \mid Z = z) &= w_z \phi(y, \tau z, (\sigma + 1)^2) + (1 - w_z) \phi(y, \tau z, \sigma^2),\\
F(y \mid Z = z) &= w_z \Phi(y, \tau z, (\sigma + 1)^2) + (1 - w_z) \Phi(y, \tau z, \sigma^2)
\end{aligned}
\end{equation*}
where $w_z = \Prob\left( X = 1 \mid Z = z\right)$. By contrast, the PDF and CDF of $Y^z$ are given by
\begin{equation*}
\begin{aligned}
f(y^z \mid Z = z) &= \frac{1}{2} \phi(y, \tau z, (\sigma + 1)^2) + \frac{1}{2} \phi(y, \tau z, \sigma^2),\\
F(Y^z \mid Z = z) &= \frac{1}{2} \Phi(y, \tau z, (\sigma + 1)^2) + \frac{1}{2} \Phi(y, \tau z, \sigma^2).
\end{aligned}
\end{equation*}
Because $X \not\independent Z$, $w_z \neq 1/2$ and therefore 
\begin{equation*}
\begin{aligned}
F^{-1}_{Y^1}(q) - F^{-1}_{Y^0}(q) &\neq F^{-1}_{Y \mid Z=1}(q) - F^{-1}_{Y \mid Z=0}(q),
\end{aligned}
\end{equation*}
as illustrated in Figure \ref{QTE}.
\begin{figure}[h]
\centering
\begin{tikzpicture}[baseline=-0.5em,scale=0.5]
	\begin{pgfonlayer}{nodelayer}
		\node [style=myvar] (1) at (-2, -1) {$Z$};
		\node [style=myvar] (3) at (0, 1) {$X$};
		\node [style=myvar] (2) at (2, -1) {$\E (Y \mid X, Z)$};
	\end{pgfonlayer}
	\begin{pgfonlayer}{edgelayer}
		\draw [style=arrow] (3) to (1);
		\draw [style=arrow] (1) to (2);
	\end{pgfonlayer}
\end{tikzpicture}\vspace{-0.6cm}
\caption{Mean causal graph}\label{graph16}
\end{figure}
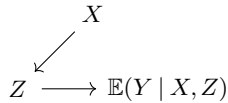

\begin{figure}
\includegraphics[width=2.5in]{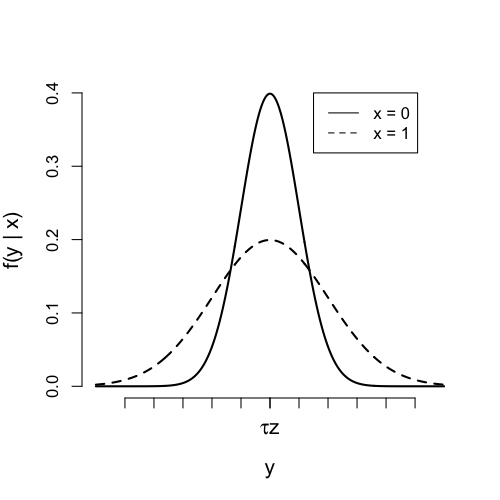}
\includegraphics[width=2.5in]{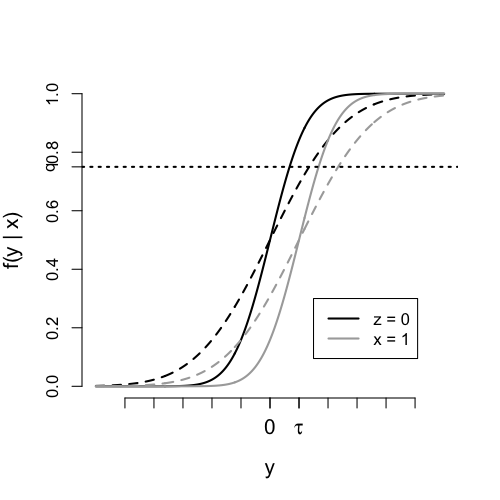}\\
\includegraphics[width=2.5in]{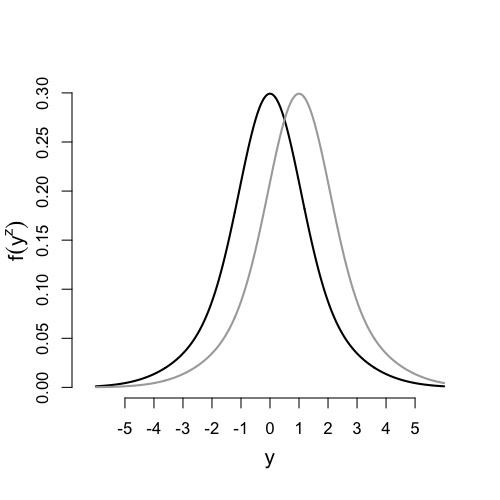}
\includegraphics[width=2.5in]{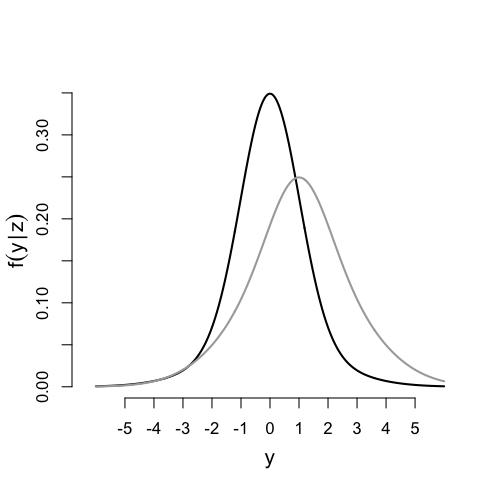}\\
\includegraphics[width=2.5in]{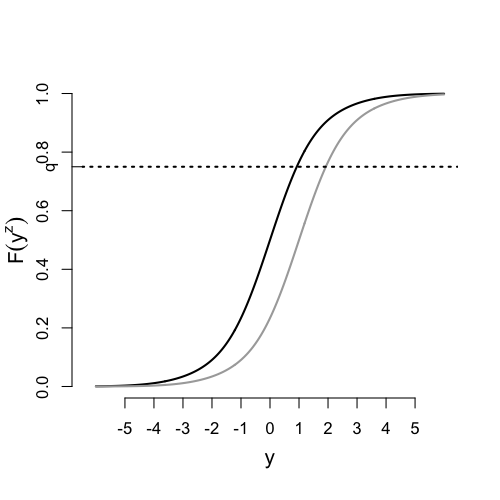}
\includegraphics[width=2.5in]{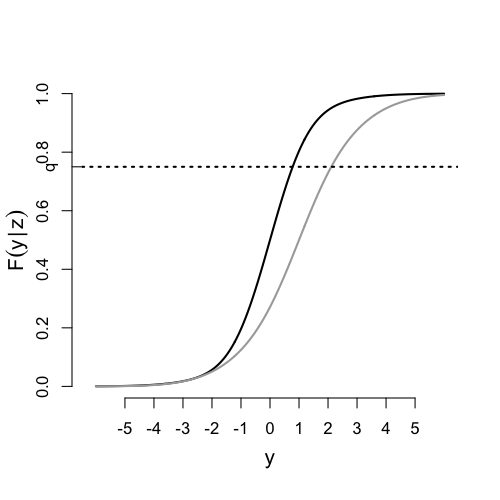}
\caption{An illustration of a confounded quantile treatment effect with unconfounded ATE. The top two panels depict the density and CDF functions of the DGP from section \ref{meanCDAG} for the four combinations of $X \in \{ 0, 1 \}$ and $Z \in \{ 0, 1 \}$. For each value of $X$ the change in the quantile is a constant shift to the right. The second row shows the densities of the potential outcome distributions and the conditional distribution of $Y \mid Z$, respectively, with $X$ integrated out. In both cases, the resulting density is a mixture of two normals with different variances and a common mean. However, the potential outcomes densities are just translations of the same mixture density, whereas the conditional distribution of $Y \mid Z$ also differs in terms of the mixture weights. The bottom row depicts the same relationship, but in terms of the CDFs. Attempts to estimate the quantile treatment effect --- shown here as the distance between the black and grey curves at the horizontal dashed line in left panel --- using the analogous distance from the right panel would misestimate the effect.}\label{QTE}
\end{figure}

\subsection{Partial randomization}\label{partial}
Some estimands require weaker assumptions than estimating the average treatment effect over the whole population does. For example, the {\em average treatment effect among the treated}, or ATT, is defined as $\E(Y^1 - Y^0 \mid Z = 1) = \E(Y^1 \mid Z = 1) - \E(Y^0 \mid Z = 1)$\footnote{In our experience, this potential outcomes notation for the ATT can give students fits, particularly the $\E(Y^0 \mid Z = 1)$ term. Such students may find the structural equation notation to be somewhat more transparent: $\E(\tau(X) \mid Z = 1)$ makes it clear that the probabilistic impact of conditioning on $Z = 1$ is to modify the distribution over $X$ defining the expectation; there is no opportunity for cognitive interference from the fact that the ``$z$'' in $Y^z$ is different from that in the condition $Z = z$.}. This estimand is important in the program evaluation literature, see for example \cite{heckman1996identification} and \cite{heckman1997matching}.

Here we use structural model notation to compare the ATT to the ATE, as relates to the ``naive'' contrast that compares the average response among treated individuals to the average response among the untreated individuals. In terms of the population, the naive contrast estimates $\E(Y \mid Z = 1) - \E(Y \mid Z = 0)$. In terms of the structural model, this is equivalent to $$\E(\mu(X) + \tau(X) \mid Z = 1) - \E(\mu(X) \mid Z = 0).$$ By definition, the exogenous errors are mean zero and vanish from the above expression. Now, randomization of $Z$ implies that $(\mu(X), \tau(X)) \independent Z$, which in turn implies that $\E(\mu(X) \mid Z = 1) = \E(\mu(X) \mid Z = 0)$ and therefore that $$\E(\mu(X) + \tau(X) \mid Z = 1) - \E(\mu(X) \mid Z = 0) = \E(\tau(X) \mid Z = 1),$$ the ATT.  Randomization further implies that $\E(\tau(X) \mid Z = 1) = \E(\tau(X))$, so that the ATE and the ATT are the same.

However, the above derivation also reveals that to estimate the ATT one only needs $\E(\mu(X) \mid Z = 1) = \E(\mu(X) \mid Z = 0)$, or what we might call {\em mean prognostic unconfoundedness}, which itself follows from $\mu(X) \independent Z$, or {\em prognostic unconfoundedness}. Thus, when the ATT is the sole interest, one only needs to rule out prognostic confounding. Meanwhile, treatment effect confounding, $\tau(X) \not \independent  Z$, entails that the ATT and ATE are different, so that the ATE remains unknown even with the ATT in hand.

Note that a similar argument works for $\E(\tau(X) \mid Z = 0)$, the average effect of the treatment on the control (untreated) population, or ATC. This is easiest to see by reparametrizing the structural model in terms of: $Z^* = 1 - Z$, $\mu^*(X) = \mu(X) + \tau(X)$, and $\tau^*(X) = -\tau(X)$. It then follows that the ATC may be estimated from the naive contrast so long as $\mu^*(X) \independent Z$.

As it relates to feature selection, it is notable that a smaller feature set may allow estimating the ATT than would be required for estimating the ATE. The following DGP is a concrete example:
\begin{align*}
X_1 \sim \mbox{Bernoulli}(1/2)&,\;\;\;X_2 \sim \mbox{Bernoulli}(1/2),\\
Z \mid X_1, X_2 &\sim \mbox{Bernoulli}(0.25 + 0.5 X_2),\\
Y \mid X_1, X_2, Z &\sim \mathcal{N}(X_1 + (1 + 2 X_2)Z, \sigma^2).
\end{align*}
In this example, $\tau(X) = \tau(X_2) = 1 + 2 X_2$, $\mu(X) = \mu(X_1) = X_1$, and the ATE is $\E(\tau(X)) = 1 + 2\E(X_2) = 2$. The ATT, on the other hand, is $\E(\tau(X) \mid Z = 1) = 1 + 2\E(X_2 \mid Z = 1) = 3/2$. It is a nice simulation exercise to demonstrate that the naive contrast is consistent for the ATT, but not the ATE.

\subsection{A two-stage estimator using two distinct control features.}\label{split_sample}
This example builds upon the ideas presented in the previous one, but returns to the goal of regression adjustments for the ATE.

Suppose we know that $\mu(X) \independent Z \mid s_1(X)$ and  $\tau(X) \independent Z \mid s_2(X)$, for distinct functions (features) $s_1$ and $s_2$. One approach to estimating the ATE under this assumption would be to stratify on the common refinement of $s_1(X)$ and $s_2(X)$, thus guaranteeing that  $(\mu(X), \tau(X)) \independent Z \mid s(X) = s_1(X) \vee s_2(X)$. But an alternative two-stage approach is possible, which requires estimating fewer individual strata means. The procedure is:
\begin{enumerate}
\item Estimate $\mu(s_1(X)) = \E(Y \mid Z = 0, s_1(X))$ from the control data. 
\item Define $R = Y - \mu(s_1(X))$. 
\item Estimate $\E(R \mid Z = 1, s_2(X))$ from the treated data. 
\item Compute the ATE as $\E_X(\E(R \mid Z = 1, s_2(X)))$, where the outer expectation is over $X$, with respect to its marginal distribution. 
\end{enumerate}
We may verify the validity of this estimator by first expressing the procedure as the following iterated expectation:
\begin{equation*}
\begin{aligned}
&\E_{X}\left( \E(Y - \E(Y \mid Z = 0, s_1(X)) \mid Z = 1, s_2(X)) \right)\\
&\quad= 
\E_{X}\left( \E(Y \mid Z = 1, s_2(X)) \right) - \E_{X}\left( \E(Y \mid Z = 0, s_1(X)) \right)\\
&\quad= \E_{X}\left( \E(\mu(X) + \tau(X) \mid Z = 1, s_2(X)) \right) - \E_{X}\left( \E(\mu(X) \mid Z = 0, s_1(X)) \right)\\
&\quad= \E_{X}\left( \E(\tau(X) \mid Z = 1, s_2(X)) \right) + \E_{X}\left( \E(\mu(X) \mid Z = 1, s_2(X)) \right) - \E_{X}\left( \E(\mu(X) \mid Z = 0, s_1(X)) \right).
\end{aligned}
\end{equation*}

By the assumption that  $\mu(X) \independent Z \mid s_1(X)$, we find that $\E(\mu(X) \mid Z = 0, s_1(X)) = \E(\mu(X) \mid Z = 1, s_1(X))$, which in turn implies that the second and third terms above are both equal to $\E_X(\mu(X) \mid Z = 1)$ (just expressed as distinct iterated expectations) and thus cancel. By the assumption that  $\tau(X) \independent Z \mid s_2(X)$, the remaining term is equal to $\E(\tau(X) \mid s_2(X))$ and the desired result follows after taking the outer expectation: $\E(\tau(X)) = \E_X(\E(\tau(X) \mid s_2(X))$.  

\section{Discussion}
To conclude, we synopsize our results and discuss further relationships to previous literature.

\subsection{Famous results or debates revisited}
The discrete covariate setting studied here allowed us to revisit several important existing results from a unique perspective. 

\subsubsection*{Virtues of the propensity score.} \cite{rosenbaum1983central} is often cited in support of propensity score methods for causal inference, but its results are often over-stated. First, there is not one propensity score, but many, one corresponding to each valid set of control features. Second, a propensity score need not be minimal; it is the minimal balancing score for the complete set of features used to create it, but balancing on those features is not necessary to estimate causal effects. Third, a propensity score method that disregards important prognostic features can be much less efficient than a method that does incorporate such features. 

\subsubsection*{Estimated versus True propensity scores.} In practice, the propensity score (corresponding to a given set of control features) is rarely known and so must be estimated. \cite{hirano2003efficient} is sometimes cited to put a positive spin on this state of affairs: estimating a propensity function is better than knowing it exactly! But the actual situation is more nuanced. The asymptotic analysis of \cite{hirano2003efficient} comparing the IPW estimator using true versus estimated propensity scores conceals the variety of specific ways the two estimators differ. Viewing the IPW as a stratification estimator in the discrete covariate setting puts these distinctions into immediate relief. One, the IPW using the true propensity scores uses different strata weights than the one using the estimated propensity scores, resulting in a higher variance estimator. Two, the IPW based on a true propensity score is able to collapse unnecessary strata, which can reduce the variance of the estimator. Three, collapsing unnecessary strata does not {\em always} reduce the variance, because the ``extraneous'' strata may be informative about {\em unconfounded} variation in the response. That is, an IPW estimator based on estimated propensity scores can have lower variance than one based on a true propensity score because it performs an implicit regression adjustment that is essentially unrelated to the propensity score. To be sure, the mathematics of \cite{hirano2003efficient} are consistent with our analysis, and one can parse their expressions for such meaning, but their analysis does not expose the importance of either variable selection or prognostic stratification. 

\subsubsection*{Regression adjustments for randomized experiments.}  \cite{freedman2008regression} is sometimes cited as a reason to avoid regression adjustment for causal effect estimation altogether. However, Freedman's result was more about model specification --- or {\em mis}specification --- than it was about regression adjustment per se. Provided that one undertakes a nonparametric adjustment, as advocated by \cite{lin2013agnostic}, Freedman's main concerns are addressed. However, nonparametric adjustment poses its own challenges, in the form of high-variance estimators. Whether or not the inclusion of strong prognostic features is enough to offset the increased variability that comes with estimating a nonparametric model with limited data is impossible to say in any generality.  Theorem \ref{theorem2} approaches this question quantitatively. 

\subsubsection*{The peril of colliders.} \cite{greenland1999causal} introduce the ``M-Graph'' and the problem of conditioning on unblocked colliders. The issue was vigorously debated in a series of articles and replies in {\it Statistics in Medicine} between 2007 and 2009. \cite{rubin2007design} suggested that all available pre-treatment covariates should be included in the conditioning set of any observational causal analysis, while others (\cite{shrier2008letter}; \cite{sjolander2009propensity}; \cite{pearl2009remarks}) contended that such a strategy could incur collider bias. 
\cite{rubin2009should} responded that unblocked colliders are a stylized problem that has few practical ramifications. This exchange in turn motivated further research, including \cite{ding2015adjust}, \cite{rohde2019bayesian}, and \cite{cinelli2020crash}. Here, we observed that should colliders appear in a set of control variables --- along with the associated blocking variables ---  regularization can unintentionally induce collider bias, revealing that colliders are not only a problem when their parents are unobserved. In particular, regularized regression approaches will struggle with colliders that are blocked by only a propensity-side ancestor. 
Additionally, Section \ref{pseudoCollider} demonstrated that composite features that combine non-collider variables can ``feature engineer'' a pseudo-collider; how likely this is to occur in practice for particular supervised learning algorithms is an interesting open question.

\subsubsection*{Conditional unconfoundedness versus mean conditional unconfoundedness.} In a discussion of \cite{angrist1996identification}, Heckman \citep{heckman1996identification} makes a point similar to the one we make in section \ref{meanCDAG}, that conditional unconfoundedness is stronger than necessary for estimating certain treatment effects. Angrist rejoins that identification based on ``functional form'' is undesirable. Here, we have taken the perspective of Heckman, as mean conditional unconfoundedness is the key notion for defining the principal deconfounding function, so it is perhaps worthwhile to unpack why. Our interest was in understanding the conditions according to which a particular set of control variables would yield a valid stratification estimator. From this perspective, a more {\em specific} assumption is {\em weaker} than a more general one: Conditional unconfoundedness implies mean conditional unconfoundedness, but not the other way around. It is the specificity of the {\em estimand} that permits the weaker (more general) assumption on the DGP. As we explored in section \ref{meanCDAG}, mean conditional unconfoundedness does not permit estimation of quantile treatment effects. In order for mean conditional unconfoundedness to license estimation of quantile treatment effects, one would need to impose additional restrictions on the DGP, such as a fixed distributional shape around the unconfounded mean. But that is not our suggestion (nor do we believe it was Heckman's).

Interestingly, this distinction between conditional unconfoundedness and mean conditional unconfoundedness is at the heart of the the difference between general causal diagrams and more traditional path analysis. By focusing on correlations, the path diagram must only respect the mean causal relationships. Sometimes this is described by saying that path analysis ``has a structural model, but no measurement model'' (Wikipedia). \\

Additionally, a number of elementary, but easily-overlooked, facts were clarified: regression, propensity score weighting (and, {\em a fortiori}, double robust estimators based thereon) are identical in the case of discrete covariates (cf. lemma \ref{ipw_strat}); CDAGs are non-unique (cf. Section \ref{transformedCDAG}), and instrumental and prognostic variable designations are inherently contingent (cf. Section \ref{noncausalSEM}).

\subsection{Methodological ecumenicalism}\label{ecumenicalism}
In section \ref{equivalence}, it was shown that the potential outcomes, CDAG, and exogenous errors definitions of conditional unconfoundedness are substantively equivalent. This result allows us to conveniently move between the conventions of these alternative frameworks, which implicitly emphasize distinct aspects of the problem they all address --- estimating treatment effects from data. 

For example, the causal graph approach reminds us that sets of valid control variables are not unique and, consequently, we must not speak of {\em the} propensity score, but rather {\em a} propensity score and, perhaps, many candidate propensity scores (cf. section \ref{propensity}). This observation is fundamental to understanding how regularization will impact bias due to feature selection on graphs including colliders and instruments.

The potential outcomes approach reminds us that the exogenous errors need not be common among the treatment arms (cf. figure \ref{graph1}). More generally, because the potential outcome notation is intrinsically individualized, it emphasizes the idea that some individuals in a population may have distinct causal diagrams; in particular, some arrows may not appear in every individual's graph. This is not at odds with the graphical formalism; rather it emerges simply because the graph alone does not fully determine the data generating process. In this paper, this distinction is not particularly important, but in estimation techniques relying on instrumental variables, it becomes critical \citep{angrist1996identification}. 

From the exogeneous errors approach, we are reminded that full conditional unconfoundedness is not actually necessary to estimate particular causal effects (cf. section \ref{meanCDAG}); we leverage this result in defining the principal deconfounding function. 

Synthesizing the three methods also clarifies common misunderstandings that can occur when operating solely within a single framework; for example, a mean regression model with exogeneous additive errors need not be structural (e.g., causal) in all of its arguments --- rather, the exogeneity of the errors narrowly licenses a causal interpretation in the treatment variable (cf. section \ref{noncausalSEM}).

\subsection{On discrete covariates with finite support}
The approach in this paper has been to consider stratification estimators in the case of discrete control variables with finite support. Discrete covariates are both common in practice (indeed, more common than continuous covariates) and pedagogically illuminating, and therefore worthy of careful study. We are aware that not everyone agrees; we read in the textbook of Imbens and Rubin (Section 12.2.2):
\small
\begin{quote}
If...we view the covariates as having a discrete
distribution with finite support, the implication of unconfoundedness is simply that one
should stratify by the values of the covariates. In that case there will be, with high probability,
in sufficiently large samples, both treated and control units with the exact same
values of the covariates. In this way we can immediately remove all biases arising from
differences between covariates, and many adjustment methods will give similar, or even
identical, answers. \\

However, as we stated before, this case rarely occurs in practice. In
many applications it is not feasible to stratify fully on all covariates, because too many
strata would have only a single unit. \\

The differences between various adjustment methods
arise precisely in such settings where it is not feasible to stratify on all values of
the covariates, and mathematically these differences are most easily analyzed in settings
with random samples from large populations using effectively continuous distributions
for the covariates...[Therefore] for the purpose of discussing various frequentist approaches to estimation and inference
under unconfoundedness...it is helpful to view the covariates as having been randomly drawn from an approximately
continuous distribution. 
\end{quote}
\normalsize

To paraphrase, the two main premises of this quote are: a) confounding --- and, more specifically, {\em de}confounding --- is relatively easy to understand in the case of discrete covariates with finite support, and b) complete stratification is infeasible in many applications. We agree with these statements. But the conclusion --- that the stylized setting of continuous covariates is therefore better suited to studying statistical methods for causal inference --- does not necessarily follow. Indeed, we employ a different stylized mathematical assumption --- that each strata has at least one treated-control contrast --- and find that, even in that case, bias variance trade-offs emerge. More importantly, these trade-offs can be studied directly, without resorting to asymptotic arguments, which may be untrustworthy guides to a method's operating characteristics in practice. For example, \cite{hahn2004functional} concludes that foreknowledge of which variables are instruments is asymptotically irrelevant for regression estimators of average treatment effects. As we have seen in Section \ref{examples} of this paper, being able to distinguish instruments from confounders is certainly relevant for finite-sample performance.

\subsection{Relationship to semi-supervised learning}
This paper considers the problem of feature selection for causal effect estimation when a propensity function is available, but a causal diagram is not. This assumption is of course implausible in many practical scenarios, although there are cases where it may be approximately true. For example, suppose that a researcher has a dataset with $n$ complete observations of $(X, Z, Y)$ and $m$ ``partially observed" samples, where $m \gg n$. Partial samples of $(X,Z)$ pairs could be used to more accurately estimate $\pi(X)$, bringing their applied problem closer to the setting studied above. Similarly, partial samples on $(X, Z = 0, Y)$ could be used to better estimate $\mu(X)$, which is particularly useful in the situation described in Section \ref{split_sample}. Such scenarios may be plausible in electronic health records, for instance, in which a treatment (say, a new blood pressure medicine) is rarely administered but an outcome (say, blood pressure) is very commonly measured. 

The idea of using large auxiliary datasets is common in machine learning, where it is known as semi-supervised learning \citep{zhu2009introduction, belkin2006manifold, liang2007use}. 
Using unlabeled data to estimate a propensity function in conjunction with machine learning or other regularization methods represent an exciting application of semi-supervised learning to the problem of causal effect estimation. 
While it is often easier to formalize and motivate the use of auxiliary data for prediction, rather than estimation, this paper shows that there is a role for function estimation techniques in machine learning causal inference.

\newpage

\section*{Acknowledgements}

This work was partially supported by NSF Grant DMS-1502640.

\bibliographystyle{imsart-nameyear} 
\bibliography{semi-supervised-propensity}

\begin{appendix}
\label{appendix-var}

\section{Proof of Theorem \ref{theorem2}} \label{appA}

We consider a sample of $n$ i.i.d. observations of $(X, Y, Z)$ from the data generating process in Equation \ref{dgp_equation}.
We assume that there exists a function $\lambda$ defined on $\mathcal{X}$ such that the ATE is identified conditional on $\lambda(X)$. 
We assume that $\lvert \lambda(X) \rvert = J < K = \lvert \mathcal{X} \rvert$ so that the unique values of $\lambda(X)$ define a non-trivial ``coarsening" of $X$. 
Consider a function $s(X)$ such that there exists at least one pair $x, x' \in \mathcal{X}$ such that 
$s(x) \neq s(x')$ while $\lambda(x) = \lambda(x')$. We assume that $s(X)$ also identifies the ATE, so that conditioning on 
$s(X)$ does not induce collider bias.

For each $j \in \lambda(\mathcal{X})$, 
there exist $m(j) \geq 1$ unique values of $s(\mathcal{X})$, which we denote as $\left\{j_1, \dots, j_m \right\}$. By the definition 
of $s(X)$, there exists at least one $j \in \lambda(\mathcal{X})$ such that $m(j) > 1$.

We define two stratification estimators of the ATE as follows:
\begin{equation*}
\begin{aligned}
\bar{\tau}^{\lambda}_{strat} &= \sum_{j \in \lambda(\mathcal{X})} \frac{N_j}{n} \left( \bar{Y}_{j,1} - \bar{Y}_{j, 0} \right)\\
N_j &= \sum_{i=1}^n \mathbf{1}\left\{\lambda(X_{i}) = j\right\}\\
N_{j, 1} &= \sum_{i=1}^n \mathbf{1}\left\{\lambda(X_{i}) = j\right\} \mathbf{1}\left(Z_{i} = 1\right\}\\
N_{j, 0} &= \sum_{i=1}^n \mathbf{1}\left\{\lambda(X_{i}) = j\right\} \mathbf{1}\left(Z_{i} = 0\right\}\\
\bar{Y}_{j,1} &= \frac{1}{N_{j, 1}} \sum_{i=1}^n Y_i \mathbf{1}\left\{\lambda(X_{i}) = j\right\} \mathbf{1}\left(Z_{i} = 1\right\}\\
\bar{Y}_{j,0} &= \frac{1}{N_{j, 0}} \sum_{i=1}^n Y_i \mathbf{1}\left\{\lambda(X_{i}) = j\right\} \mathbf{1}\left(Z_{i} = 0\right\}
\end{aligned}
\;\;\;\;\;
\begin{aligned}
\bar{\tau}^{s}_{strat} &= \sum_{j \in \lambda(\mathcal{X})} \left( \sum_{\ell=1}^{m(j)} \frac{N_{j\ell}}{n} \left( \bar{Y}_{j\ell,1} - \bar{Y}_{j\ell, 0} \right)\right)\\
N_{j\ell} &= \sum_{i=1}^n \mathbf{1}\left\{s(X_{i}) = j_{\ell}\right\}\\
N_{j\ell, 1} &= \sum_{i=1}^n \mathbf{1}\left\{s(X_{i}) = j_{\ell}\right\} \mathbf{1}\left\{Z_{i} = 1\right\}\\
N_{j\ell, 0} &= \sum_{i=1}^n \mathbf{1}\left\{s(X_{i}) = j_{\ell}\right\} \mathbf{1}\left\{Z_{i} = 0\right\}\\
\bar{Y}_{j\ell,1} &= \frac{1}{N_{j\ell, 1}} \sum_{i=1}^n Y_{i} \mathbf{1}\left\{s(X_{i}) = j_{\ell}\right\} \mathbf{1}\left(Z_{i} = 1\right\}\\
\bar{Y}_{j\ell,0} &= \frac{1}{N_{j\ell, 0}} \sum_{i=1}^n Y_{i} \mathbf{1}\left\{s(X_{i}) = j_{\ell}\right\} \mathbf{1}\left(Z_{i} = 0\right\}
\end{aligned}
\end{equation*}

Now, we consider a $j \in \lambda(\mathcal{X})$ with $m(j) > 1$.
We introduce some notation. 
\begin{equation*}
\begin{aligned}
\mu_{j,1} &= \E\left(Y \mid \lambda(X) = j, Z = 1\right)\\
\mu_{j,0} &= \E\left(Y \mid \lambda(X) = j, Z = 0\right)\\
\sigma^2_{j,1} &= \V\left(Y \mid \lambda(X) = j, Z = 1\right)\\
\sigma^2_{j,0} &= \V\left(Y \mid \lambda(X) = j, Z = 0\right)\\
\end{aligned}\;\;\;\;\;
\begin{aligned}
\mu_{j\ell,1} &= \E\left(Y \mid s(X) = j\ell, Z = 1\right)\\
\mu_{j\ell,0} &= \E\left(Y \mid s(X) = j\ell, Z = 0\right)\\
\sigma^2_{j\ell,1} &= \V\left(Y \mid s(X) = j\ell, Z = 1\right)\\
\sigma^2_{j\ell,0} &= \V\left(Y \mid s(X) = j\ell, Z = 0\right)\\
\end{aligned}
\end{equation*}

By the law of iterated expectations and the law of total variance, it follows that 
\begin{equation*}
\begin{aligned}
\mu_{j,1} &= \E\left(\E\left(Y \mid s(X) = j, Z = 1\right) \mid \lambda(X) = j, Z = 1\right) = \E\left(\mu_{j\ell,1} \mid \lambda(X) = j, Z = 1\right)\\
\sigma^2_{j,1} &= \E\left(\V\left(Y \mid s(X) = j, Z = 1\right) \mid \lambda(X) = j, Z = 1\right) + \V\left(\E\left(Y \mid s(X) = j, Z = 1\right) \mid \lambda(X) = j, Z = 1\right)\\
&= \E\left(\sigma^2_{j\ell,1} \mid \lambda(X) = j, Z = 1\right) + \V\left(\mu_{j\ell, 1} \mid \lambda(X) = j, Z = 1\right)
\end{aligned}
\end{equation*}

We denote
\begin{equation*}
\begin{aligned}
\bar{\mu}_{j, 1} &= \E\left(\mu_{j\ell,1} \mid \lambda(X) = j, Z = 1\right) = \mu_{j, 1}\\
\bar{\sigma}^2_{j, 1} &= \E\left(\sigma^2_{j\ell,1} \mid \lambda(X) = j, Z = 1\right)\\
\mathbf{v}\left( \mu_{j\ell, 1} \right) &= \V\left(\mu_{j\ell, 1} \mid \lambda(X) = j, Z = 1\right)
\end{aligned}
\end{equation*}

Conditioning on $\mathbf{N} = \left\{N_{j1, 1}, \dots, N_{jm, 1}, N_{j1, 0}, \dots, N_{jm, 0}\right\}$, we see
\begin{equation*}
\begin{aligned}
\V\left(\sum_{\ell=1}^{m(j)} N_{j\ell} \bar{Y}_{j\ell,1} \mid \mathbf{N} \right) &= \sum_{\ell=1}^{m(j)} N_{j\ell}^2 \V\left(\bar{Y}_{j\ell,1}\right) = \sum_{\ell=1}^{m(j)} N_{j\ell}^2 \frac{\sigma^2_{j\ell,1}}{N_{j\ell, 1}} = \sum_{\ell=1}^{m(j)} \frac{(N_{j\ell, 1}+N_{j\ell, 0})^2}{N_{j\ell, 1}} \sigma^2_{j\ell,1}\\
\V\left(N_j \bar{Y}_{j,1} \mid \mathbf{N} \right) &= N_j^2 \V\left(\bar{Y}_{j,1}\right) = N_j^2 \frac{\sigma^2_{j\ell,1}}{N_{j,1}} = \frac{(N_{j,1}+N_{j,0})^2}{N_{j,1}} \sigma^2_{j\ell,1} 
= \frac{(\sum_{\ell=1}^{m(j)} (N_{j\ell,1}+N_{j\ell,0}))^2}{\sum_{\ell=1}^{m(j)} N_{j\ell,1}} \sigma^2_{j\ell,1} \\ 
&= \frac{(\sum_{\ell=1}^{m(j)} (N_{j\ell,1}+N_{j\ell,0}))^2}{\sum_{\ell=1}^{m(j)} N_{j\ell,1}} \left(\bar{\sigma}^2_{j, 1} + \mathbf{v}\left( \mu_{j\ell, 1} \right) \right)\\
\end{aligned}
\end{equation*}

and

\begin{equation*}
\begin{aligned}
\E\left(\sum_{\ell=1}^{m(j)} N_{j\ell} \bar{Y}_{j\ell,1} \mid \mathbf{N} \right) &= \sum_{\ell=1}^{m(j)} N_{j\ell} \E\left(\bar{Y}_{j\ell,1}\right) = \sum_{\ell=1}^{m(j)} N_{j\ell} \E\left(Y \mid s(X) = j\ell, Z = 1\right) = \sum_{\ell=1}^{m(j)} N_{j\ell} \mu_{j\ell,1}\\
\E\left(N_j \bar{Y}_{j,1} \mid \mathbf{N} \right) &= N_j \E\left(\bar{Y}_{j,1}\right) = N_j \E\left(Y \mid \lambda(X) = j, Z = 1\right) = N_j \mu_{j, 1} = \left(\sum_{\ell=1}^{m(j)} N_{j\ell,1}\right) \bar{\mu}_{j, 1}\\ 
\end{aligned}
\end{equation*}

Thus, we have that 
\begin{equation*}
\begin{aligned}
\V\left(\sum_{\ell=1}^{m(j)} N_{j\ell} \bar{Y}_{j\ell,1}\right) &= \E\left(\V\left(\sum_{\ell=1}^{m(j)} N_{j\ell} \bar{Y}_{j\ell,1} \mid \mathbf{N} \right)\right) + \V\left(\E\left(\sum_{\ell=1}^{m(j)} N_{j\ell} \bar{Y}_{j\ell,1} \mid \mathbf{N} \right)\right)\\
&= \E\left(\sum_{\ell=1}^{m(j)} \frac{(N_{j\ell, 1}+N_{j\ell, 0})^2}{N_{j\ell, 1}} \sigma^2_{j\ell,1}\right) + \V\left(\sum_{\ell=1}^{m(j)} N_{j\ell} \mu_{j\ell,1}\right)\\
\V\left(N_j \bar{Y}_{j,1}\right) &= \E\left(\V\left(N_j \bar{Y}_{j,1} \mid \mathbf{N} \right)\right) + \V\left(\E\left(N_j \bar{Y}_{j,1} \mid \mathbf{N} \right)\right)\\
&= \E\left(\frac{(\sum_{\ell=1}^{m(j)} (N_{j\ell,1}+N_{j\ell,0}))^2}{\sum_{\ell=1}^{m(j)} N_{j\ell,1}} \left(\bar{\sigma}^2_{j, 1} + \mathbf{v}\left( \mu_{j\ell, 1} \right) \right)\right) + \V\left(\left(\sum_{\ell=1}^{m(j)} N_{j\ell,1}\right) \bar{\mu}_{j, 1}\right)
\end{aligned}
\end{equation*}

We can broaden this to entire set of observations:
\begin{equation*}
\begin{aligned}
\V\left(\bar{\tau}^{s}_{strat}\right) &= \V\left(\sum_{j \in \lambda(\mathcal{X})} \sum_{\ell=1}^{m(j)} \frac{N_{j\ell}}{n} \left( \bar{Y}_{j\ell,1} - \bar{Y}_{j\ell,0} \right)\right)\\
&= \frac{1}{n^2} \left[ \E\left( \sum_{j \in \lambda(\mathcal{X})} \sum_{\ell=1}^{m(j)} (N_{j\ell, 1}+N_{j\ell, 0})^2 \left(\frac{\sigma^2_{j\ell,1}}{N_{j\ell, 1}} + \frac{\sigma^2_{j\ell,0}}{N_{j\ell, 0}} \right) \right) + \V\left(\sum_{j \in \lambda(\mathcal{X})} \sum_{\ell=1}^{m(j)} N_{j\ell} \left( \mu_{j\ell,1} - \mu_{j\ell,0}\right)\right) \right]\\
\V\left(\bar{\tau}^{\lambda}_{strat}\right) &= \V\left(\sum_{j \in \lambda(\mathcal{X})} \frac{N_j}{n} \left( \bar{Y}_{j,1} - \bar{Y}_{j,0} \right) \right)\\
&= \frac{1}{n^2}\E\left(\sum_{j \in \lambda(\mathcal{X})} \left(\sum_{\ell=1}^{m(j)} (N_{j\ell,1}+N_{j\ell,0})\right)^2 \left( \frac{\bar{\sigma}^2_{j, 1} + \mathbf{v}\left( \mu_{j\ell, 1} \right)}{\sum_{\ell=1}^{m(j)} N_{j\ell,1}} + \frac{\bar{\sigma}^2_{j, 0} + \mathbf{v}\left( \mu_{j\ell, 0} \right)}{\sum_{\ell=1}^{m(j)} N_{j\ell,0}} \right)\right)\\
&\;\;\;\;\; + \frac{1}{n^2}\V\left(\sum_{j \in \lambda(\mathcal{X})}\left(\sum_{\ell=1}^{m(j)} N_{j\ell}\right) \left( \bar{\mu}_{j, 1} - \bar{\mu}_{j, 0} \right) \right)
\end{aligned}
\end{equation*}

\subsection{Case I: equal sub-strata means and variances}  \label{appA1}

If $\sigma^2_{j\ell,i} = \bar{\sigma}^2_{j,i}$ and $\mu_{j\ell,i} = \bar{\mu}_{j,i}$ for all $\ell, i \in \left\{1, \dots, m(j)\right\} \times \left\{0, 1\right\}$, then $\mathbf{v}\left( \mu_{j\ell, i} \right) = 0$ and the variance and expectation 
terms factor out of both expressions and we are left to compare the nonlinear sums of the strata cell sizes. 
Focusing on $\lambda(X) = j$ and $Z = 1$, we show by induction that 
$\sum_{\ell=1}^{m(j)} \frac{(N_{j\ell, 1}+N_{j\ell, 0})^2}{N_{j\ell, 1}} \geq \frac{(\sum_{\ell=1}^{m(j)} (N_{j\ell,1}+N_{j\ell,0}))^2}{\sum_{\ell=1}^{m(j)} N_{j\ell,1}}$ for 
positive cell sizes $N_{j\ell, 1}$.

For the base case, suppose that $m(j) = 2$ so that $\sum_{\ell=1}^{m(j)} \frac{(N_{j\ell, 1}+N_{j\ell, 0})^2}{N_{j\ell, 1}} = \frac{(N_{j1, 1}+N_{j1, 0})^2}{N_{j1, 1}} + \frac{(N_{j2, 1}+N_{j2, 0})^2}{N_{j2, 1}}$ 
and $\frac{(\sum_{\ell=1}^{m(j)} (N_{j\ell,1}+N_{j\ell,0}))^2}{\sum_{\ell=1}^{m(j)} N_{j\ell,1}} = \frac{((N_{j1,1}+N_{j1,0}) + (N_{j2,1}+N_{j2,0}))^2}{N_{j1,1} + N_{j2,1}}$.
For ease of exposition, we let
\begin{equation*}
\begin{aligned}
a &= N_{j_1, 1}\\
c &= N_{j_2, 1}
\end{aligned}\;\;\;\;\;\;\;
\begin{aligned}
b &= N_{j_1, 0}\\
d &= N_{j_2, 0}
\end{aligned}
\end{equation*}
and we thus compare $\frac{(a + b)^2}{a} + \frac{(c + d)^2}{c}$ with $\frac{(a + b + c + d)^2}{a + c}$ where $a,b,c,d > 0$

\begin{equation*}
\begin{aligned}
0 &\leq \left[c(a+b) - a(c+d)\right]^2\\
0 &\leq c^2(a+b)^2 + a^2 (c+d)^2 - 2ac(a+b)(c+d)\\
2ac(a+b)(c+d) &\leq c^2(a+b)^2 + a^2 (c+d)^2\\
ac(a+b)^2 + 2ac(a+b)(c+d) + ac(c+d)^2 &\leq ac(a+b)^2 + c^2(a+b)^2 + a^2 (c+d)^2 + ac(c+d)^2\\
ac\left[(a+b) + (c+d)\right]^2 &\leq \left[c(a+b)^2 + a(c+d)^2\right] (a+c)\\
\frac{\left[(a+b) + (c+d)\right]^2}{(a+c)} &\leq \frac{\left[c(a+b)^2 + a(c+d)^2\right]}{ac}\\
\frac{\left[(a+b) + (c+d)\right]^2}{(a+c)} &\leq \frac{(a+b)^2}{a} + \frac{(c+d)^2}{c}
\end{aligned}
\end{equation*}

Now, we proceed to the induction case. Assume that $\sum_{\ell=1}^{m(j)} \frac{(N_{j\ell, 1}+N_{j\ell, 0})^2}{N_{j\ell, 1}} \geq \frac{(\sum_{\ell=1}^{m(j)} (N_{j\ell,1}+N_{j\ell,0}))^2}{\sum_{\ell=1}^{m(j)} N_{j\ell,1}}$. 
We consider a new stratum, indexed $m(j) + 1$ and we see that 
\begin{equation*}
\begin{aligned}
\sum_{\ell=1}^{m(j)} \frac{(N_{j\ell, 1}+N_{j\ell, 0})^2}{N_{j\ell, 1}} &\geq \frac{(\sum_{\ell=1}^{m(j)} (N_{j\ell,1}+N_{j\ell,0}))^2}{\sum_{\ell=1}^{m(j)} N_{j\ell,1}}\\
\sum_{\ell=1}^{m(j)} \frac{(N_{j\ell, 1}+N_{j\ell, 0})^2}{N_{j\ell, 1}} + \frac{(N_{m(j)+1, 1}+N_{m(j)+1, 0})^2}{N_{m(j)+1, 1}} &\geq \frac{(\sum_{\ell=1}^{m(j)} (N_{j\ell,1}+N_{j\ell,0}))^2}{\sum_{\ell=1}^{m(j)} N_{j\ell,1}} + \frac{(N_{m(j)+1, 1}+N_{m(j)+1, 0})^2}{N_{m(j)+1, 1}}\\
\sum_{\ell=1}^{m(j)+1} \frac{(N_{j\ell, 1}+N_{j\ell, 0})^2}{N_{j\ell, 1}} &\geq \frac{(\sum_{\ell=1}^{m(j)} (N_{j\ell,1}+N_{j\ell,0}))^2}{\sum_{\ell=1}^{m(j)} N_{j\ell,1}} + \frac{(N_{m(j)+1, 1}+N_{m(j)+1, 0})^2}{N_{m(j)+1, 1}}
\end{aligned}
\end{equation*}

Letting
\begin{equation*}
\begin{aligned}
a &= \sum_{\ell=1}^{m(j)} N_{j\ell,1}\\
c &= N_{m(j)+1, 1}
\end{aligned}\;\;\;\;\;\;\;
\begin{aligned}
b &= \sum_{\ell=1}^{m(j)} N_{j\ell,0}\\
d &= N_{m(j)+1, 0}
\end{aligned}
\end{equation*}
we see that 
\begin{equation*}
\begin{aligned}
\frac{(\sum_{\ell=1}^{m(j)} (N_{j\ell,1}+N_{j\ell,0}))^2}{\sum_{\ell=1}^{m(j)} N_{j\ell,1}} + \frac{(N_{m(j)+1, 1}+N_{m(j)+1, 0})^2}{N_{m(j)+1, 1}} = \frac{(a+b)^2}{a} + \frac{(c+d)^2}{c} &\geq \frac{(a+b + c + d)^2}{a+c}\\
&= \frac{(\sum_{\ell=1}^{m(j)+1} (N_{j\ell,1}+N_{j\ell,0}))^2}{\sum_{\ell=1}^{m(j)+1} N_{j\ell,1}}
\end{aligned}
\end{equation*}
and the relationship follows by induction.

Thus, when $\E\left(Y \mid s(X) = j\ell, Z=1\right)$ and $\V\left(Y \mid s(X) = j\ell, Z=1\right)$ are constant for all $j \in \lambda(\mathcal{X})$ and $\ell \in \left\{1, \dots, m(j)\right\}$, 
it follows that $\V\left( \bar{\tau}^{\lambda}_{strat} \right) \leq \V\left( \bar{\tau}^{s}_{strat} \right)$. 

We let
\begin{equation*}
\begin{aligned}
\alpha_1 &= \frac{1}{n^2} \E\left( \sum_{j \in \lambda(\mathcal{X})} \bar{\sigma}^2_{j,1} \left( \left( \sum_{\ell=1}^{m(j)} \frac{(N_{j\ell, 1}+N_{j\ell, 0})^2}{N_{j\ell, 1}} \right) - \left( \frac{ \left( \sum_{\ell=1}^{m(j)} (N_{j\ell,1}+N_{j\ell,0})\right)^2}{\sum_{\ell=1}^{m(j)} N_{j\ell,1}}  \right) \right) \right)\\
\alpha_0 &= \frac{1}{n^2} \E\left( \sum_{j \in \lambda(\mathcal{X})} \bar{\sigma}^2_{j,0} \left( \left( \sum_{\ell=1}^{m(j)} \frac{(N_{j\ell, 1}+N_{j\ell, 0})^2}{N_{j\ell, 0}} \right) - \left( \frac{ \left( \sum_{\ell=1}^{m(j)} (N_{j\ell,1}+N_{j\ell,0})\right)^2}{\sum_{\ell=1}^{m(j)} N_{j\ell,0}}  \right) \right) \right)
\end{aligned}
\end{equation*}
so that $\alpha = \alpha_1 + \alpha_0$ is the degree to which $\V\left( \bar{\tau}^{\lambda}_{strat} \right) \leq \V\left( \bar{\tau}^{s}_{strat} \right)$ when all substrata of $s(X)$ are constant. We see that this 
depends on $\bar{\sigma}^2_{j,1}$, $\bar{\sigma}^2_{j,0}$, and the distribution of $N_{j\ell, i}$ for each $j$ and $i$.

\subsection{Case II: unequal strata means or variances} \label{appA2}

We partition $\lambda(\mathcal{X})$ into three sets: 
\begin{itemize}
\item $\mathcal{A}$: $m(a) = 1$ for all $a \in \mathcal{A}$
\item $\mathcal{B}$: for all $b \in \mathcal{B}$, either
\begin{itemize}
\item $m(b) > 1$
\item $\sigma^2_{b\ell,i}$ and $\mu_{b\ell,i}$ are constant for all $\ell \in m(b)$, $i \in \left\{ 0, 1 \right\}$
\end{itemize}
\item $\mathcal{C}$: For all $c \in \mathcal{C}$
\begin{itemize}
\item $m(c) > 1$, and
\item $\sigma^2_{c\ell,i}$ or $\mu_{c\ell,i}$ is non-constant for some $i \in \left\{ 0, 1 \right\}$
\end{itemize}
\end{itemize}

In the previous section $\mathcal{C} = \varnothing$, so that the variance comparison is uncomplicated: $s(X)$ was ``overstratified" relative to $\lambda(X)$ and 
as a result $\V\left( \bar{\tau}^{\lambda}_{strat} \right) \leq \V\left( \bar{\tau}^{s}_{strat} \right)$.

In this case, $\mathcal{C} \neq \varnothing$ so that there may be variance reduction to stratification (see \cite{lohr2019sampling} for one reference). We define several terms
\begin{equation*}
\begin{aligned}
\beta_1 &= \frac{1}{n^2} \E\left( \sum_{b \in \mathcal{B}} \bar{\sigma}^2_{b,1} \left( \left( \sum_{\ell=1}^{m(b)} \frac{(N_{b\ell, 1}+N_{b\ell, 0})^2}{N_{b\ell, 1}} \right) - \left( \frac{ \left( \sum_{\ell=1}^{m(b)} (N_{b\ell,1}+N_{b\ell,0})\right)^2}{\sum_{\ell=1}^{m(b)} N_{b\ell,1}}  \right) \right) \right)\\
\beta_0 &= \frac{1}{n^2} \E\left( \sum_{b \in \mathcal{B}} \bar{\sigma}^2_{b,0} \left( \left( \sum_{\ell=1}^{m(b)} \frac{(N_{b\ell, 1}+N_{b\ell, 0})^2}{N_{b\ell, 0}} \right) - \left( \frac{ \left( \sum_{\ell=1}^{m(b)} (N_{b\ell,1}+N_{b\ell,0})\right)^2}{\sum_{\ell=1}^{m(b)} N_{b\ell,0}}  \right) \right) \right)\\
c_1 &= \frac{1}{n^2}\E\left(\sum_{c \in \mathcal{C}} \left(\sum_{\ell=1}^{m(c)} (N_{c\ell,1}+N_{c\ell,0})\right)^2 \left( \frac{\bar{\sigma}^2_{c, 1} + \mathbf{v}\left( \mu_{c\ell, 1} \right)}{\sum_{\ell=1}^{m(c)} N_{c\ell,1}}\right)\right) + \frac{1}{n^2}\V\left(\sum_{c \in \mathcal{C}}\left(\sum_{\ell=1}^{m(c)} N_{c\ell}\right) \left( \bar{\mu}_{c, 1} \right)\right)\\
c_0 &= \frac{1}{n^2}\E\left(\sum_{c \in \mathcal{C}} \left(\sum_{\ell=1}^{m(c)} (N_{c\ell,1}+N_{c\ell,0})\right)^2 \left( \frac{\bar{\sigma}^2_{c, 0} + \mathbf{v}\left( \mu_{c\ell, 0} \right)}{\sum_{\ell=1}^{m(c)} N_{c\ell,0}} \right)\right) + \frac{1}{n^2}\V\left(\sum_{c \in \mathcal{C}}\left(\sum_{\ell=1}^{m(c)} N_{c\ell}\right) \left( - \bar{\mu}_{c, 0} \right)\right)\\
c_2 &= \frac{1}{n^2} \left[ \E\left( \sum_{c \in \mathcal{C}} \sum_{\ell=1}^{m(c)} (N_{c\ell, 1}+N_{c\ell, 0})^2 \left(\frac{\sigma^2_{c\ell,1}}{N_{c\ell, 1}} + \frac{\sigma^2_{c\ell,0}}{N_{c\ell, 0}} \right) \right) + \V\left(\sum_{c \in \mathcal{C}} \sum_{\ell=1}^{m(c)} N_{c\ell} \left( \mu_{c\ell,1} - \mu_{c\ell,0}\right)\right) \right]\\
\eta &= c_1 + c_0 - c_2\\
\nu &= \beta_1 + \beta_0
\end{aligned}
\end{equation*}

We see that $\V\left( \bar{\tau}^{\lambda}_{strat} \right) > \V\left( \bar{\tau}^{s}_{strat} \right)$ if $\eta > \nu$, 
where $\eta$ refers to the reduction in variance by stratifying on $s(X)$ within $\mathcal{C}$ and $\nu$ refers to 
the increase in variance by ``over-stratifying" on $\mathcal{B}$.

\section{Proof of Theorem \ref{theorem3}} \label{appB}

We consider a sample of $n$ i.i.d. observations of $(X, Y, Z)$ from the data generating process in Equation \ref{dgp_equation}.
Let $\pi(x) = \Prob\left( Z = 1 \mid X = x \right)$ refer to the ``true propensity function" and $\hat{\pi}(x) = N_{x,1} / N_x$ refer to the ``empirical propensity function." 

We define two stratification estimators of the ATE as follows:
\begin{equation*}
\begin{aligned}
\bar{\tau}^{\hat{\pi}}_{strat} &= \sum_{x \in \mathcal{X}} \frac{N_x}{n} \left( \bar{Y}_{x,1} - \bar{Y}_{x, 0} \right)
\end{aligned}
\;\;\;\;\;
\begin{aligned}
\bar{\tau}^{\pi}_{strat} &= \sum_{x \in \mathcal{X}} \frac{N_x}{n} \left( \frac{\hat{\pi}(x)}{\pi(x)} \bar{Y}_{x,1} - \frac{1-\hat{\pi}(x)}{1-\pi(x)} \bar{Y}_{x,0} \right)
\end{aligned}
\end{equation*}
where $N_x$, $N_{x,1}$, $N_{x,0}$, $\bar{Y}_{x,1}$, and $\bar{Y}_{x,0}$ are defined similar to Appendix \ref{appA}:
\begin{equation*}
\begin{aligned}
N_x &= \sum_{i=1}^n \mathbf{1} \left(X_i = x\right)\\
N_{x,1} &= \sum_{i=1}^n \mathbf{1} \left(X_i = x\right) \mathbf{1} \left(Z_i = 1\right)\\
N_{x,0} &= \sum_{i=1}^n \mathbf{1} \left(X_i = x\right) \mathbf{1} \left(Z_i = 0\right)\\
\bar{Y}_{x,1} &= \frac{1}{N_{x,1}} \sum_{i=1}^n \mathbf{1} Y_i \left(X_i = x\right) \mathbf{1} \left(Z_i = 1\right)\\
\bar{Y}_{x,0} &= \frac{1}{N_{x,0}} \sum_{i=1}^n \mathbf{1} Y_i \left(X_i = x\right) \mathbf{1} \left(Z_i = 0\right)
\end{aligned}
\end{equation*}

Now, we consider an arbitrary $x \in \mathcal{X}$ with $N_{x,1} > 0$ and $N_{x,0} > 0$.
We introduce some notation. 
\begin{equation*}
\begin{aligned}
\mu_{x,1} &= \E\left(Y \mid X = x, Z = 1\right)\\
\mu_{x,0} &= \E\left(Y \mid X = x, Z = 0\right)\\
\sigma^2_{x,1} &= \V\left(Y \mid X = x, Z = 1\right)\\
\sigma^2_{x,0} &= \V\left(Y \mid X = x, Z = 0\right)\\
\end{aligned}
\end{equation*}

Conditioning on $\mathbf{N} = \left\{N_{x, 1}, N_{x, 0}: x \in \mathcal{X}\right\}$, we see
\begin{equation*}
\begin{aligned}
\V\left(\sum_{x \in \mathcal{X}} N_{x} \bar{Y}_{x,1} \mid \mathbf{N} \right) &= \sum_{x \in \mathcal{X}} N_{x}^2 \V\left(\bar{Y}_{x,1}\right) = \sum_{x \in \mathcal{X}} N_{x}^2 \frac{\sigma^2_{x,1}}{N_{x, 1}} = \sum_{x \in \mathcal{X}} \frac{(N_{x,1} + N_{x,0})^2}{N_{x,1}} \sigma^2_{x,1}\\
\V\left(\sum_{x \in \mathcal{X}} N_{x} \frac{\hat{\pi}(x)}{\pi(x)} \bar{Y}_{x,1} \mid \mathbf{N} \right) &= \sum_{x \in \mathcal{X}} N_{x}^2 \left(\frac{\hat{\pi}(x)}{\pi(x)}\right)^2 \V\left(\bar{Y}_{x,1}\right) = \sum_{x \in \mathcal{X}} N_{x}^2 \left(\frac{\hat{\pi}(x)}{\pi(x)}\right)^2 \frac{\sigma^2_{x,1}}{N_{x, 1}}\\
& = \sum_{x \in \mathcal{X}} \left(\frac{\hat{\pi}(x)}{\pi(x)}\right)^2 \frac{(N_{x,1} + N_{x,0})^2}{N_{x,1}} \sigma^2_{x,1}\\
\end{aligned}
\end{equation*}

and

\begin{equation*}
\begin{aligned}
\E\left(\sum_{x \in \mathcal{X}} N_{x} \bar{Y}_{x,1} \mid \mathbf{N} \right) = \sum_{x \in \mathcal{X}} N_{x} \E\left(\bar{Y}_{x,1}\right) &= \sum_{x \in \mathcal{X}} N_{x} \E\left(Y \mid X = x, Z = 1\right) = \sum_{x \in \mathcal{X}} N_{x} \mu_{x,1},\\
\E\left(\sum_{x \in \mathcal{X}} N_{x} \frac{\hat{\pi}(x)}{\pi(x)} \bar{Y}_{x,1} \mid \mathbf{N} \right) &= \sum_{x \in \mathcal{X}} N_{x} \frac{\hat{\pi}(x)}{\pi(x)} \E\left(\bar{Y}_{x,1}\right)\\
& = \sum_{x \in \mathcal{X}} \frac{\hat{\pi}(x)}{\pi(x)} N_{x} \E\left(Y \mid X = x, Z = 1\right) \\
&= \sum_{x \in \mathcal{X}} \frac{\hat{\pi}(x)}{\pi(x)} N_{x} \mu_{x,1}.
\end{aligned}
\end{equation*}

Thus, we have that 
\begin{equation*}
\begin{aligned}
\V\left(\sum_{x \in \mathcal{X}} N_{x} \bar{Y}_{x,1}\right) &= \E\left(\V\left(\sum_{x \in \mathcal{X}} N_{x} \bar{Y}_{x,1} \mid \mathbf{N} \right)\right) + \V\left(\E\left(\sum_{x \in \mathcal{X}} N_{x} \bar{Y}_{x,1} \mid \mathbf{N} \right)\right)\\
&= \E\left(\sum_{x \in \mathcal{X}} \frac{(N_{x,1} + N_{x,0})^2}{N_{x,1}} \sigma^2_{x,1}\right) + \V\left(\sum_{x \in \mathcal{X}} N_{x} \mu_{x,1}\right)\\
\V\left(\sum_{x \in \mathcal{X}} N_{x} \frac{\hat{\pi}(x)}{\pi(x)} \bar{Y}_{x,1}\right) &= \E\left(\V\left(\sum_{x \in \mathcal{X}} N_{x} \frac{\hat{\pi}(x)}{\pi(x)} \bar{Y}_{x,1} \mid \mathbf{N} \right)\right) + \V\left(\E\left(\sum_{x \in \mathcal{X}} N_{x} \frac{\hat{\pi}(x)}{\pi(x)} \bar{Y}_{x,1} \mid \mathbf{N} \right)\right)\\
&= \E\left(\sum_{x \in \mathcal{X}} \left(\frac{\hat{\pi}(x)}{\pi(x)}\right)^2 \frac{(N_{x,1} + N_{x,0})^2}{N_{x,1}} \sigma^2_{x,1}\right) + \V\left(\sum_{x \in \mathcal{X}} \frac{\hat{\pi}(x)}{\pi(x)} N_{x} \mu_{x,1}\right)
\end{aligned}
\end{equation*}

We can broaden this to entire set of observations:
\begin{equation*}
\begin{aligned}
\V\left(\bar{\tau}^{\hat{\pi}}_{strat}\right) &= \V\left(\sum_{x \in \mathcal{X}} \frac{N_x}{n} \left( \bar{Y}_{x,1} - \bar{Y}_{x, 0} \right)\right)\\
&= \frac{1}{n^2} \left[  \E\left(\sum_{x \in \mathcal{X}} (N_{x,1} + N_{x,0})^2 \left(\frac{\sigma^2_{x,1}}{N_{x,1}} + \frac{\sigma^2_{x,0}}{N_{x,0}} \right)\right) + \V\left(\sum_{x \in \mathcal{X}} N_{x} \left( \mu_{x,1} - \mu_{x,0} \right)\right) \right]\\
&= \frac{1}{n^2} \left[  \E\left(\sum_{x \in \mathcal{X}} \left(\frac{N_{x,1}}{\left(\hat{\pi}(x)\right)^2} \sigma^2_{x,1} + \frac{N_{x,0}}{\left(1-\hat{\pi}(x)\right)^2} \sigma^2_{x,0} \right)\right) + \V\left(\sum_{x \in \mathcal{X}} N_x \left( \frac{\hat{\pi}(x)}{\hat{\pi}(x)} \mu_{x,1} - \frac{1-\hat{\pi}(x)}{1-\hat{\pi}(x)}\mu_{x,0} \right) \right) \right]\\
\V\left(\bar{\tau}^{\pi}_{strat}\right) &= \V\left(\sum_{x \in \mathcal{X}} \frac{N_x}{n} \left( \frac{\hat{\pi}(x)}{\pi(x)} \bar{Y}_{x,1} - \frac{1-\hat{\pi}(x)}{1-\pi(x)} \bar{Y}_{x, 0} \right)\right)\\
&= \frac{1}{n^2} \E\left(\sum_{x \in \mathcal{X}} (N_{x,1} + N_{x,0})^2 \left( \left(\frac{\hat{\pi}(x)}{\pi(x)}\right)^2 \frac{\sigma^2_{x,1}}{N_{x,1}} + \left(\frac{1-\hat{\pi}(x)}{1-\pi(x)}\right)^2 \frac{\sigma^2_{x,0}}{N_{x,0}} \right)\right)\\
&\;\;\;\;\;\; + \frac{1}{n^2} \V\left(\sum_{x \in \mathcal{X}} N_{x} \left( \frac{\hat{\pi}(x)}{\pi(x)} \mu_{x,1} - \frac{1-\hat{\pi}(x)}{1-\pi(x)} \mu_{x,0} \right)\right) \\
&= \frac{1}{n^2} \E\left(\sum_{x \in \mathcal{X}} \frac{N_{x,1}}{\left(\pi(x)\right)^2} \sigma^2_{x,1} + \frac{N_{x,0}}{\left(1-\pi(x)\right)^2} \sigma^2_{x,0} \right)\\
&\;\;\;\;\;\; + \frac{1}{n^2} \V\left(\sum_{x \in \mathcal{X}} N_{x} \left( \frac{\hat{\pi}(x)}{\pi(x)} \mu_{x,1} - \frac{1-\hat{\pi}(x)}{1-\pi(x)} \mu_{x,0} \right)\right) \\
\end{aligned}
\end{equation*}

\subsection{Case I: $\mu(x) = \tau(x) = 0$} \label{appB1}

First, consider the degenerate case in which $\mu(x) = \tau(x) = 0$ for all $x \in \mathcal{X}$. In this case, the second term vanishes 
and we are left to compare $\E\left(\sum_{x \in \mathcal{X}} \left(\frac{N_{x,1}}{\left(\hat{\pi}(x)\right)^2} \sigma^2_{x,1} + \frac{N_{x,0}}{\left(1-\hat{\pi}(x)\right)^2} \sigma^2_{x,0} \right)\right)$ 
with $\E\left(\sum_{x \in \mathcal{X}} \frac{N_{x,1}}{\left(\pi(x)\right)^2} \sigma^2_{x,1} + \frac{N_{x,0}}{\left(1-\pi(x)\right)^2} \sigma^2_{x,0} \right)$. 
Since we have defined $\hat{\pi}$ empirically, we have that $\E\left(\hat{\pi}(x)\right) = \pi(x)$. 
We also have, by iterated expectations that $\E(N_{x,1}) = \E(\E(N_{x,1} \mid N_{x})) = \E( N_x \pi(x) ) = \pi(x) \E( N_x )$. 

Now, define $g(x,y) = y^2 / x$, so that $g(N_{x,1}, N_{x}) = N_x^2 / N_{x,1} = N_{x,1} / \hat{\pi}(x)^2$.
$g(N_{x,1}, N_{x})$ can be shown to be a convex function as its Hessian is positive semidefinite (\cite{boyd2004convex}). 
Thus, by Jensen's inequality, 
\begin{equation*}
\begin{aligned}
\E\left(g(N_{x,1}, N_{x}) \right) &\geq g\left(\E\left(N_{x,1}, N_{x} \right)\right) \\
\E\left( \frac{N_x^2}{N_{x,1}} \right) &\geq g\left(\pi(x) N_x, N_{x} \right)\\
\E\left( \frac{N_{x,1}}{\hat{\pi}(x)^2} \right) &\geq \frac{N_x}{\pi(x)} = \frac{N_x\pi(x)}{\pi(x)^2} = \frac{\E\left(N_{x,1}\right)}{\pi(x)^2} = \E\left(\frac{N_{x,1}}{\pi(x)^2}\right)
\end{aligned}
\end{equation*}

Intuitively, this result shows that when there is ``nothing to estimate" in that $y$ is degenerate, the variance of the $\bar{\tau}^{\hat{\pi}}_{strat}$ estimator is 
greater than that of the $\bar{\tau}^{\pi}_{strat}$ estimator. 

\subsection{Case II: $\lvert \mu(x) \rvert + \lvert \tau(x) \rvert > \epsilon$ for some $x \in \mathcal{X}$ and for a data-dependent $\epsilon$} \label{appB2}

Setting aside this degenerate case, we first define 
\begin{equation*}
\begin{aligned}
a &= \E\left(\sum_{x \in \mathcal{X}} \left(\frac{N_{x,1}}{\left(\hat{\pi}(x)\right)^2} \sigma^2_{x,1} + \frac{N_{x,0}}{\left(1-\hat{\pi}(x)\right)^2} \sigma^2_{x,0} \right)\right) - 
\E\left(\sum_{x \in \mathcal{X}} \frac{N_{x,1}}{\left(\pi(x)\right)^2} \sigma^2_{x,1} + \frac{N_{x,0}}{\left(1-\pi(x)\right)^2} \sigma^2_{x,0} \right)
\end{aligned}
\end{equation*}

$a$ is nonnegative as demonstrated above. We assume that there exists at least one $x$ for which either or both of $\mu_{x,1} \neq 0$ and $\mu_{x,0} \neq 0$ is true.
Without loss of generality, we focus on the case in which there is \textit{exactly one} such $x$ (i.e., $\mu_{x,1} \neq 0$ and / or $\mu_{x,0} \neq 0$, while $\mu_{x',1} = \mu_{x',0} = 0$ for all $x' \in \mathcal{X} \setminus x$)
\begin{equation*}
\begin{aligned}
\V\left(\sum_{x \in \mathcal{X}} N_x \left( \frac{\hat{\pi}(x)}{\pi(x)} \mu_{x,1} - \frac{1-\hat{\pi}(x)}{1-\pi(x)}\mu_{x,0} \right) \right) &= \V\left(N_x \left( \frac{\hat{\pi}(x)}{\pi(x)} \mu_{x,1} - \frac{1-\hat{\pi}(x)}{1-\pi(x)}\mu_{x,0} \right) \right) \\
&= \V\left(N_x \left( \frac{\hat{\pi}(x)}{\pi(x)} \mu_{x,1} - \frac{1-\hat{\pi}(x)}{1-\pi(x)}\mu_{x,0} \right)\right)\\
&= \V\left(N_x \left( \frac{N_{x,1}}{N_x\pi(x)} \mu_{x,1} - \frac{N_{x,0}}{N_x(1-\pi(x))}\mu_{x,0} \right)\right)\\
&= \V\left(\frac{N_{x,1}}{\pi(x)} \mu_{x,1} - \frac{N_{x,0}}{(1-\pi(x))}\mu_{x,0} \right)\\
&= \V\left(\frac{N_{x,1}}{\pi(x)} \mu_{x,1} - \frac{N_{x,0}}{(1-\pi(x))}\mu_{x,0} \right)\\
\V\left(\sum_{x \in \mathcal{X}} N_x \left( \mu_{x,1} - \mu_{x,0}\right) \right) &= \V\left( N_x \left( \mu_{x,1} - \mu_{x,0}\right) \right)\\
\end{aligned}
\end{equation*}

We rewrite $\mu_{x,1} = \left(\mu_{x,1}-\mu_{x,0}\right) + \mu_{x,0}$ and see that 
\begin{equation*}
\begin{aligned}
\V\left(\sum_{x \in \mathcal{X}} N_x \left( \frac{\hat{\pi}(x)}{\pi(x)} \mu_{x,1} - \frac{1-\hat{\pi}(x)}{1-\pi(x)}\mu_{x,0} \right) \right) &= \V\left(\frac{N_{x,1}}{\pi(x)} \mu_{x,1} - \frac{N_{x,0}}{(1-\pi(x))}\mu_{x,0} \right)\\
&= \V\left(\frac{N_{x,1}}{\pi(x)} \left[ \left(\mu_{x,1}-\mu_{x,0}\right) + \mu_{x,0} \right] - \frac{N_{x,0}}{(1-\pi(x))}\mu_{x,0} \right)\\
&= \V\left(\frac{N_{x,1}}{\pi(x)} \left(\mu_{x,1}-\mu_{x,0}\right) + \mu_{x,0} \left(\frac{N_{x,1}}{(\pi(x))} - \frac{N_{x,0}}{(1-\pi(x))}\right) \right)
\end{aligned}
\end{equation*}

We know that $\V\left(\frac{N_{x,1}}{\pi(x)} \left(\mu_{x,1}-\mu_{x,0}\right) + \mu_{x,0} \left(\frac{N_{x,1}}{(\pi(x))} - \frac{N_{x,0}}{(1-\pi(x))}\right) \right)$
is the sum of the variances of the two terms plus twice their covariance, so we focus on 
$\V\left(\frac{N_{x,1}}{\pi(x)} \left(\mu_{x,1}-\mu_{x,0}\right) \right)$. 

\begin{equation*}
\begin{aligned}
\V\left(\frac{N_{x,1}}{\pi(x)} \left(\mu_{x,1}-\mu_{x,0}\right) \right) &= \E \left( \V\left(\frac{N_{x,1}}{\pi(x)} \left(\mu_{x,1}-\mu_{x,0}\right) \mid N_x \right) \right) + \V\left( \E \left( \frac{N_{x,1}}{\pi(x)} \left(\mu_{x,1}-\mu_{x,0}\right) \mid N_x \right)\right)\\
&= \E\left(N_x \frac{\left(\mu_{x,1} - \mu_{x,0}\right)^2\left(1 - \pi(x)\right)}{\pi(x)}\right) + \V \left( N_x \left(\mu_{x,1} - \mu_{x,0}\right) \right)
\end{aligned}
\end{equation*}

It follows that both
\begin{equation*}
\begin{aligned}
\E\left(N_x \frac{\left(\mu_{x,1} - \mu_{x,0}\right)^2\left(1 - \pi(x)\right)^2}{\pi(x)}\right) &\geq 0\\
\V\left(\mu_{x,0} \left(\frac{N_{x,1}}{(\pi(x))} - \frac{N_{x,0}}{(1-\pi(x))}\right) \right) &\geq 0
\end{aligned}
\end{equation*}

Similarly,
\begin{equation*}
\begin{aligned}
\V\left(\mu_{x,0} \left(\frac{N_{x,1}}{(\pi(x))} - \frac{N_{x,0}}{(1-\pi(x))}\right) \right) &= \E \left( \V\left(\mu_{x,0} \left(\frac{N_{x,1}}{(\pi(x))} - \frac{N_{x,0}}{(1-\pi(x))}\right) \mid N_x \right) \right)\\
&\;\;\;\;\; + \V\left( \E \left(\mu_{x,0} \left(\frac{N_{x,1}}{(\pi(x))} - \frac{N_{x,0}}{(1-\pi(x))}\right) \mid N_x \right)\right)\\
&= \E \left( \V\left(\mu_{x,0} \left(\frac{N_{x,1}}{\pi(x)} - \frac{N_{x,0}}{(1-\pi(x))}\right) \mid N_x \right) \right)\\
&= \E \left( \V\left(\mu_{x,0} \left(\frac{N_{x,1}}{\pi(x)} - \frac{N_x - N_{x,1}}{(1-\pi(x))}\right) \mid N_x \right) \right)\\
&= \E \left( \V\left(\mu_{x,0} \left(N_{x,1}\left(\frac{1}{\pi(x)} + \frac{1}{(1-\pi(x))}\right) - \frac{N_x}{(1-\pi(x))} \right) \mid N_x \right) \right)\\
&= \E \left( \V\left(\mu_{x,0} \left(N_{x,1}\left(\frac{1}{\pi(x)} + \frac{1}{(1-\pi(x))}\right) \right) \mid N_x \right) \right)\\
&= \E \left( \V\left(\mu_{x,0} \left(N_{x,1}\left(\frac{1}{\pi(x)(1-\pi(x))}\right) \right) \mid N_x \right) \right)\\
&= \E \left( \left(\frac{\mu_{x,0}}{\pi(x)(1-\pi(x))}\right)^2 \V\left( \left(N_{x,1} \right) \mid N_x \right) \right)\\
&= \E \left( \left(\frac{\mu_{x,0}}{\pi(x)(1-\pi(x))}\right)^2 \left( \pi(x)(1-\pi(x)) \left(N_{x} \right) \right) \right)\\
&= \E \left( \frac{\mu_{x,0}^2}{\pi(x)(1-\pi(x))} N_{x} \right) \\
\end{aligned}
\end{equation*}
so we are left to analyze 
\begin{equation*}
\begin{aligned}
\mbox{Cov}\left(\frac{N_{x,1}}{\pi(x)}, \frac{N_{x,1}}{(\pi(x))} - \frac{N_{x,0}}{(1-\pi(x))} \right)
\end{aligned}
\end{equation*}

We can show that this is nonnegative using the law of total covariance (see for example \cite{casella2002statistical})
\begin{equation*}
\begin{aligned}
\mbox{Cov}\left(\frac{N_{x,1}}{\pi(x)}, \frac{N_{x,1}}{(\pi(x))} - \frac{N_{x,0}}{(1-\pi(x))} \right) &= \E\left( \mbox{Cov}\left(\frac{N_{x,1}}{\pi(x)}, \frac{N_{x,1}}{(\pi(x))} - \frac{N_{x,0}}{(1-\pi(x))} \mid N_x \right) \right)\\
& + \mbox{Cov}\left( \E\left(\frac{N_{x,1}}{\pi(x)} \mid N_x \right) , \E\left(\frac{N_{x,1}}{(\pi(x))} - \frac{N_{x,0}}{(1-\pi(x))} \mid N_x \right) \right)\\
&= \E\left( \mbox{Cov}\left(\frac{N_{x,1}}{\pi(x)}, \frac{N_{x,1}}{(\pi(x))} - \frac{N_{x,0}}{(1-\pi(x))} \mid N_x \right) \right)\\
&= \E\left( \E \left(\frac{N_{x,1}^2}{\pi(x)^2} - \frac{N_{x,1}N_{x,0}}{\pi(x)(1-\pi(x))} \mid N_x \right) \right)
\end{aligned}
\end{equation*}

and so we solve
\begin{equation*}
\E\left( \E \left(\frac{N_{x,1}^2}{\pi(x)^2} - \frac{N_{x,1}N_{x,0}}{\pi(x)(1-\pi(x))} \mid N_x \right) \right) = \E\left( \frac{\pi(x)N_x + \pi(x)^2 N_x (N_x - 1)}{\pi(x)^2} - \E \left( \frac{N_{x,1}N_{x,0}}{\pi(x)(1-\pi(x))} \mid N_x \right) \right)
\end{equation*}
\begin{equation*}
\begin{aligned}
\hspace{2in}&= \E\left( \frac{\pi(x)(1-\pi(x))N_x + \pi(x)^2 N_x^2}{\pi(x)^2} - \E \left( \frac{N_{x,1}N_{x,0}}{\pi(x)(1-\pi(x))} \mid N_x \right) \right)\\
&= \E\left( \frac{\pi(x)(1-\pi(x))N_x + \pi(x)^2 N_x^2}{\pi(x)^2} - \E \left( \frac{N_{x,1}(N_x-N_{x,1})}{\pi(x)(1-\pi(x))} \mid N_x \right) \right)\\
&= \E\left( \frac{\pi(x)(1-\pi(x))N_x + \pi(x)^2 N_x^2}{\pi(x)^2} - \E \left( \frac{N_xN_{x,1}-N_{x,1}^2}{\pi(x)(1-\pi(x))} \mid N_x \right) \right)\\
&= \E\left( \frac{\pi(x)(1-\pi(x))N_x + \pi(x)^2 N_x^2}{\pi(x)^2(1-\pi(x))} - \frac{N_x^2}{1-\pi(x)} \right)\\
&= \E\left( \frac{\pi(x)(1-\pi(x))N_x}{\pi(x)^2(1-\pi(x))} + \frac{N_x^2}{(1-\pi(x))} - \frac{N_x^2}{1-\pi(x)} \right)\\
&= \E\left( \frac{N_x}{\pi(x)} \right)\\
\end{aligned}
\end{equation*}

Now, consider the entire variance
\begin{equation*}
\begin{aligned}
&\V\left(\frac{N_{x,1}}{\pi(x)} \left(\mu_{x,1}-\mu_{x,0}\right) + \mu_{x,0} \left(\frac{N_{x,1}}{(\pi(x))} - \frac{N_{x,0}}{(1-\pi(x))}\right) \right) \\
&\;\;\;\;\;= \V\left(\frac{N_{x,1}\left(\mu_{x,1}-\mu_{x,0}\right)}{\pi(x)}\right) \\
&\;\;\;\;\;+ \V\left(\mu_{x,0}\left(\frac{N_{x,1}}{(\pi(x))} - \frac{N_{x,0}}{(1-\pi(x))}\right) \right) \\
&\;\;\;\;\;+ 2\mu_{x,0}\left(\mu_{x,1}-\mu_{x,0}\right) \mbox{Cov}\left(\frac{N_{x,1}}{\pi(x)}, \frac{N_{x,1}}{(\pi(x))} - \frac{N_{x,0}}{(1-\pi(x))} \right)\\
&\;\;\;\;\;= \E\left(N_x \frac{\left(\mu_{x,1} - \mu_{x,0}\right)^2\left(1 - \pi(x)\right)}{\pi(x)}\right) + \V \left( N_x \left(\mu_{x,1} - \mu_{x,0}\right) \right) \\
&\;\;\;\;\;+ \E \left( \frac{\mu_{x,0}^2}{\pi(x)(1-\pi(x))} N_{x} \right) \\
&\;\;\;\;\;+ 2\mu_{x,0}\left(\mu_{x,1}-\mu_{x,0}\right) \E\left( \frac{N_x}{\pi(x)} \right)\\
&\;\;\;\;\;= \V \left( N_x \left(\mu_{x,1} - \mu_{x,0}\right) \right)\\
&\;\;\;\;\;+ \left(\mu_{x,1} - \mu_{x,0}\right)^2\left(1 - \pi(x)\right) \E\left(\frac{N_x}{\pi(x)}\right) + \frac{\mu_{x,0}^2}{(1-\pi(x))} \E \left( \frac{N_{x}}{\pi(x)} \right) + 2\mu_{x,0}\left(\mu_{x,1}-\mu_{x,0}\right) \E\left( \frac{N_x}{\pi(x)} \right)\\
&\;\;\;\;\;= \V \left( N_x \left(\mu_{x,1} - \mu_{x,0}\right) \right) + \E\left(\frac{N_x}{\pi(x)}\right) \left( \left(\mu_{x,1} - \mu_{x,0}\right)\sqrt{1 - \pi(x)} + \frac{\mu_{x,0}}{\sqrt{1-\pi(x)}} \right)^2
\end{aligned}
\end{equation*}

Thus we have that 
\begin{equation*}
\begin{aligned}
\V\left(\sum_{x \in \mathcal{X}} N_x \left( \frac{\hat{\pi}(x)}{\pi(x)} \mu_{x,1} - \frac{1-\hat{\pi}(x)}{1-\pi(x)}\mu_{x,0} \right) \right) \geq \V\left(\sum_{x \in \mathcal{X}} N_x \left( \mu_{x,1} - \mu_{x,0}\right) \right)
\end{aligned}
\end{equation*}

Now, we define 
\begin{equation*}
\begin{aligned}
b &= \V\left(\sum_{x \in \mathcal{X}} N_x \left( \frac{\hat{\pi}(x)}{\pi(x)} \mu_{x,1} - \frac{1-\hat{\pi}(x)}{1-\pi(x)}\mu_{x,0} \right) \right) - \V\left(\sum_{x \in \mathcal{X}} N_x \left( \mu_{x,1} - \mu_{x,0}\right) \right)
\end{aligned}
\end{equation*}

and we see that $\V\left(\bar{\tau}^{\hat{\pi}}_{strat}\right) \leq \V\left(\bar{\tau}^{\pi}_{strat}\right)$ when $b \geq a$. 
This is true when $\mu_{x,1}$ and / or $\mu_{x,0}$ are ``far enough" from zero that the variance benefit of estimating 
propensity scores outweighs the numerical effect detailed in Appendix \ref{appB1}. Formally, we may choose a data-dependent $\epsilon$ 
such that if $\lvert\mu(x)\rvert + \lvert\tau(x)\rvert > \epsilon$ for at least one $x \in \mathcal{X}$, then $b \geq a$.

\section{Proof of Corollary \ref{corollary1}} \label{appC}

The proof follows directly from the proof outlined in Appendix \ref{appA1}. 
To make the extension to Theorem \ref{theorem2} as clear as possible, we index level sets of $\pi(\mathcal{X})$ using $j$. 
For each $j \in \pi(\mathcal{X})$, we let $m(j) = \lvert \left\{ x : \pi(x) = j \right\} \rvert$. 
By the assumptions of Corollary \ref{corollary1}, we have that for all $j \in \pi(\mathcal{X})$ and all $\ell = \left\{1, \dots, m(j)\right\}$,
\begin{itemize}
\item $\sigma^2_{j\ell} = \sigma_j^2$
\item $\mu_{j\ell, 1} = \bar{\mu}_{j, 1}$
\item $\mu_{j\ell, 0} = \bar{\mu}_{j, 0}$
\end{itemize}
where the second two points are true because both $\mu(x)$ and $\tau(x)$ are constant for each $x$ with $\pi(x) = j$. Thus, it follows that 
$$\V \left( \sum_{j \in s(\mathcal{X})} \bar{\tau}_{\textrm{ipw}}^{\hat{\pi},s} \right) \leq \V \left( \sum_{j \in s(\mathcal{X})} \left( \sum_{\pi: g(\pi) = j} \frac{N_{\pi}}{N_j} \bar{\tau}_{\textrm{ipw}}^{\hat{\pi},\pi} \right) \right)$$

\section{Proof of Corollary \ref{corollary2}} \label{appD}

As in Appendix \ref{appC}, the proof follows more or less directly from the proof of Theorem \ref{theorem2}. 
By assumption, there exist $x, x' \in \mathcal{X}$ such that 
\begin{itemize}
\item $x \neq x'$
\item $\pi(x) \neq \pi(x')$
\item $s(x) = s(x')$
\item $\mu(x) = \mu(x')$
\item $\tau(x) = \tau(x')$
\end{itemize}
We let $j = s(x) = s(x')$ denote a level set of $s$ for which the above conditions are true, 
and we define $m(j) = \lvert \left\{ x : \pi(x) = j \right\} \rvert$. 
By the assumptions of Corollary \ref{corollary2}, we have that for all $\ell = \left\{1, \dots, m(j)\right\}$,
\begin{itemize}
\item $\sigma^2_{j\ell} = \sigma_j^2$
\item $\mu_{j\ell, 1} = \bar{\mu}_{j, 1}$
\item $\mu_{j\ell, 0} = \bar{\mu}_{j, 0}$
\end{itemize}
Thus, it follows that 
$$\V \left( \sum_{j \in s(\mathcal{X})} \bar{\tau}_{\textrm{ipw}}^{\hat{\pi},s} \right) \leq \V \left( \sum_{j \in s(\mathcal{X})} \left( \sum_{\pi: g(\pi) = j} \frac{N_{\pi}}{N_j} \bar{\tau}_{\textrm{ipw}}^{\hat{\pi},\pi} \right) \right)$$

\end{appendix}
\end{document}